\documentclass[11pt,a4paper]{article}
\usepackage{jheppub}
\usepackage[english]{babel}
\usepackage{graphicx}
\usepackage{framed}
\usepackage{physics}

\DeclareMathOperator*{\SumInt}{%
\mathchoice%
  {\ooalign{$\displaystyle\sum$\cr\hidewidth$\displaystyle\int$\hidewidth\cr}}
  {\ooalign{\raisebox{.14\height}{\scalebox{.7}{$\textstyle\sum$}}\cr\hidewidth$\textstyle\int$\hidewidth\cr}}
  {\ooalign{\raisebox{.2\height}{\scalebox{.6}{$\scriptstyle\sum$}}\cr$\scriptstyle\int$\cr}}
  {\ooalign{\raisebox{.2\height}{\scalebox{.6}{$\scriptstyle\sum$}}\cr$\scriptstyle\int$\cr}}
}

\usepackage{dsfont}
\usepackage{listings}
\usepackage{systeme}
\usepackage{ytableau}

\usepackage{xcolor}

\definecolor{codegreen}{rgb}{0,0.6,0}
\definecolor{codegray}{rgb}{0.5,0.5,0.5}
\definecolor{codepurple}{rgb}{0.58,0,0.82}
\definecolor{backcolour}{rgb}{0.95,0.95,0.92}
\definecolor{light-gray}{gray}{0.95} 

\lstdefinestyle{mystyle}{
    backgroundcolor=\color{light-gray},   
    commentstyle=\color{codegreen},
    keywordstyle=\color{magenta},
    numberstyle=\tiny\color{codegray},
    stringstyle=\color{codepurple},
    basicstyle=\ttfamily\footnotesize,
    breakatwhitespace=false,         
    breaklines=false,   
    captionpos=b,   
    keepspaces=true,
    numbersep=5pt,                  
    showspaces=false,                
    showstringspaces=false,
    showtabs=false,                  
    tabsize=2
}

\usepackage{latexsym}
\usepackage[normalem]{ulem}
\usepackage{amsmath}
\usepackage{amsthm}
\usepackage{amssymb}
\usepackage{amsfonts}
\usepackage{enumerate}
\usepackage[utf8]{inputenc}
\usepackage{color}

\theoremstyle{plain}
\newtheorem{theorem}{Theorem}
\newtheorem{corollary}{Corollary}

\newtheorem*{remark2}{Remark}
\theoremstyle{remark}
\newtheorem*{remark}{Example}

\title{The invariant space of multi-Higgs doublet models}
\author{M. P. Bento}
\affiliation{CFTP, Departamento de F\'{i}sica, Instituto Superior T\'{e}cnico,
Universidade de Lisboa,\\
Avenida Rovisco Pais 1, 1049 Lisboa, Portugal}

\emailAdd{miguel.pedra.bento@tecnico.ulisboa.pt}

\abstract{In a model with more than one scalar doublet, the parameter space encloses
   both physical and unphysical information. Invariant theory provides a detailed
   description of the counting and characterization of the physical parameter space.
   The Hilbert series for the 3HDM is computed for the first time
   using partition analysis, in particular Omega calculus, giving rise to the
   possibility of a full description of its physical parameters. A rigorous
   counting of the physical parameters is given for the full
   class of models with $N$ scalar doublets as well as a decomposition of the Lagrangian
   into irreducible representations of $\mathrm{SU}(N)$. For the first time
   we derive a basis-invariant technique for counting parameters in a Lagrangian
   with both basis-invariant redundancies and global symmetries.}

\begin{document}
\maketitle

\section{Introduction}

In high-energy physics symmetries play a fundamental role in
model building of both theory and phenomenology multi-scalar theories. 
Nevertheless, we are often interested on the invariants of these symmetries as they
may relate to gauge invariance, physical parameters or the construction of the Lagrangian.

The scalar potential has been thoroughly studied in the physics literature,
where various basis-invariant methods 
\cite{Botella:1994cs, Ginzburg:2004vp, Gunion:2005ja, Davidson:2005cw} and
more group-theoretical methods such as billinears
\cite{Nagel, Maniatis:2006fs, Maniatis:2006jd, Maniatis:2007vn, Nishi:2006tg, Nishi:2007nh, Nishi:2007dv, Ivanov:2005hg, Ivanov:2006yq, Degee:2009vp}
were used. Recent work has developed
a new perspective on the group structure of the parameter space by using invariant theory.
The characterization of invariants has also been instrumental
to study physical parameters
and CP violation. The Hilbert series of the 2HDM was
first obtained in \cite{Bednyakov:2018cmx} and later
studied in the context of CP violation in \cite{Trautner:2018ipq}. 
With this technique,
the complete roadmap to the basis-invariant
description of the 2HDM built on invariant theory was achieved in \cite{Bento:2020jei}.

Invariant theory, a field of algebraic geometry, concerns the study of precisely these invariants
and was developed by many prominent mathematicians such as David Hilbert, Emmy Noether and
physicist Hermann Weyl.
It has
been used in the context of string theory \cite{Benvenuti:2006qr}. 
More recently, an excellent review
of these methods was given in \cite{Lehman:2015via} along with strategies for
handling couplings with derivatives in EFTs.
As a theory, it also provides a framework for a full group-theoretical
perspective of the parameter space.

As complete as it may be, invariant theory relies heavily on the computation
of a formal quantity, known as the Hilbert series. In it lies the full count
and characterization of any physical parameter in a theory. A shortcoming of
this strategy is the sometimes insurmountable
calculation of a very large number of residues of multivariate
integrals. Here, we introduce a technique developed by P. MacMahon \cite{Macmahon},
later called Omega calculus. With it, many challenging complex integrals
become attainable, as we will show with the very complicated case of the 3HDM.

By using several results obtained throughout the years in the mathematical literature,
we extract properties for the class of NHDM (N Higgs Doublet models). 
In particular,
we show that the Hilbert series
is not needed to compute the number of physical parameters in multi-Higgs doublet models.

With the knowledge on how the vector space of a Lagrangian decomposes in irreducible
representations of  any group, we derive a technique that counts
the number of parameters in a Lagrangian with
both basis-invariant redundancies and global symmetries. This method
does not require knowledge of invariant theory, only of the group structure
of the symmetry group $G$.

\section{Group structure of the scalar potential}\label{sec:group_s}

There are essentially two perspectives regarding the group structure of the scalar
sector of a NHDM: the fields and their representations \cite{Nishi:2006tg, Ivanov:2005hg},
and the parameter space
and its representation \cite{Bednyakov:2018cmx,Trautner:2018ipq}, both under a basis
transformation group. We will follow the latter by decomposing the parameter
space into irreducible representations of $\mathrm{SU}(N)$.
Assigning this structure
will allow us to build the invariant analogue of the parameter space under
basis transformations. In other words, build the physical parameter
space of any multi-Higgs doublet model.

In the most general NHDM, the Lagrangian potential can be written as
\begin{equation}\label{eq:NHDM_lag}
    V_{H} = \mu_{ij} (\Phi^\dagger_i \Phi_j) + z_{ij,kl} (\Phi^\dagger_i \Phi_j) 
    (\Phi^\dagger_k \Phi_l) \, ,
\end{equation}
where the matrices follow hermiticity and symmetry properties
$\mu_{ij}=\mu_{ji}^*$ and $z_{ij,kl} = z_{kl,ij} = z_{ji, lk}^*$. It is well known
that eq.~\eqref{eq:NHDM_lag} is not unique and that we can always perform a basis
transformation $\mathrm{SU}(N)$ such that we generate the same physical
theory. Furthermore, the fields transform under the fundamental
representation $\mathbf{r}_f$. Thus, $\mu_{ij}$ and $z_{ij,kl}$ transform as
\begin{align}
    &\mu_{i j} \rightarrow \bar{\mathbf{r}}_f  \, \otimes \, \mathbf{r}_f \, , \nonumber \\[2mm]
    &z_{ij,kl} \rightarrow \mathrm{Sym} \left(  \mathbf{r}_f  \, \otimes \,
    \mathbf{\mathbf{r}}_f
    \, \otimes \, \bar{\mathbf{r}}_f  \, \otimes \, \bar{\mathbf{r}}_f \right) \, ,
\end{align}
where $\mathrm{Sym}$ denotes essentially the symmetry property of $z_{ij,kl}$. With
this decomposition we find
\begin{align}\label{eq:dimV_sym}
    \mathrm{Sym} \left( \mathbf{r}_f  \, \otimes \,
    \mathbf{\mathbf{r}}_f
    \, \otimes \, \bar{\mathbf{r}}_f  \, \otimes \, \bar{\mathbf{\mathbf{r}}}_f \right)
    &= \mathrm{Sym} \left( \left[\mathrm{Sym}^2 (\mathbf{r}_f)  \, \oplus \, 
    \mathrm{Alt}^2 (\mathbf{r}_f)
    \right] \, \otimes \, 
    \left[\mathrm{Sym}^2 (\bar{\mathbf{r}}_f )  \, \oplus \, \mathrm{Alt}^2 (\bar{\mathbf{r}}_f )
    \right] \right) \nonumber \\[2mm]
    &= \left[ \mathrm{Sym}^2 (\mathbf{r}_f) \,
    \otimes \, \mathrm{Sym}^2 (\bar{\mathbf{r}}_f ) \right] \, \oplus \, 
    \left[ \mathrm{Alt}^2 (\mathbf{r}_f) \, \otimes \, \mathrm{Alt}^2 (\bar{\mathbf{r}}_f )  \right]
    \, ,
\end{align}
where we follow the well known group theory result $\mathbf{r} \, \otimes \, \mathbf{r}
= \mathrm{Sym}^2(\mathbf{r}) \, \oplus \, \mathrm{Alt}^2(\mathbf{r})$,
where $\mathrm{Sym}$ and $\mathrm{Alt}$ are respectively the symmetric and antisymmetric parts
of the tensor product. In section~\ref{sec:symmetries} we will revisit eq.~\eqref{eq:dimV_sym}.

The results above point to a full decomposition of the parameter space of the scalar
potential in terms of irreducible representations of $\mathrm{SU}(N)$.
Thus, we define the vector space of the parameters as $V$, defined as the space
that transforms with
\begin{equation}\label{eq:V_decomposition}
    V = \mu \, \oplus \, z \rightarrow \bar{\mathbf{r}}_f  \, \otimes \, \mathbf{r}_f \, \oplus \,
    \left[ \mathrm{Sym}^2 (\mathbf{r}_f) \,
    \otimes \, \mathrm{Sym}^2 (\bar{\mathbf{r}}_f ) \right] \, \oplus \, 
    \left[ \mathrm{Alt}^2 (\mathbf{r}_f) \, \otimes \, \mathrm{Alt}^2 (\bar{\mathbf{r}}_f )  \right]
    \, .
\end{equation}
The dimension $\dim V$ is the number of parameters in the potential. This can be readily computed
from eq.~\eqref{eq:V_decomposition}. With $\dim \mathbf{r}_f = N$, the number of doublets,
and
\begin{align}
    \dim \left[ \mathrm{Sym}^2 (\mathbf{r}_f) \right] &= \frac{N(N+1)}{2} \, , \nonumber \\[2mm]
    \dim \left[ \mathrm{Alt}^2 (\mathbf{r}_f) \right] &= \frac{N(N-1)}{2} \, ,
\end{align}
we finally get
\begin{equation}
    \dim V = N^2 + \left( \frac{N(N+1)}{2} \right)^2 + \left( \frac{N(N-1)}{2} \right)^2
    =
    \frac{N^2(N^2+3)}{2} \, ,
\end{equation}
as the number of parameters in the NHDM \cite{Nishi:2006tg, Ferreira:2008zy, Ivanov:2010ww, Bento:2017eti}.
\begin{remark}
In the 2HDM we may decompose
\begin{equation}
    \mu_{ij} \rightarrow \mathbf{2} \, \otimes \, \mathbf{2} = \mathbf{1} \, \oplus \, \mathbf{3}
    \, ,
\end{equation}
and similarly we decompose
\begin{align}
    z_{ij,kl} \rightarrow \mathrm{Sym} \left(  \mathbf{2}  \, \otimes \,
    \mathbf{2}
    \, \otimes \, \mathbf{2} \, \otimes \, \mathbf{2} \right)
    &= [ \mathbf{3} \, \otimes \, \mathbf{3}] \, \oplus \,
    [\mathbf{1} \, \otimes \, \mathbf{1}] \nonumber \\[2mm]
    &= 2 (\mathbf{1}) \, \oplus \, \mathbf{3} \, \oplus \, \mathbf{5} \, .
\end{align}
Then the vector space of parameters $V$ transforms as
\begin{equation}\label{eq:V_decomposition_2hdm}
    V \rightarrow 3(\mathbf{1}) \, \oplus \, 2(\mathbf{3}) \, \oplus \, \mathbf{5} \, ,
\end{equation}
and $\dim V = 14$, the number of parameters in the 2HDM.
\end{remark}

The decomposition of $V$ is instrumental to the analysis of the physical parameters
of the NHDM. Throughout this paper we won't concern ourselves with the structure of the
representation themselves, as it is not needed for any of our results. Nevertheless,
a systematic approach to this calculation is given in ref.~\cite{Trautner:2018ipq}
with the use of projectors.

\section{The invariant space formalism}

The decomposition of $V$ into irreducible representations provides a framework
for how a group $G$ acts on $V$. Nevertheless,
no physical parameter is given in experiment in matrix form. The computation
of physical parameters is often built on contracted tensors for which the
answer is a number, a polynomial in the Lagrangian parameters \cite{Botella:1994cs}.
Thus, it is important to introduce the notion of polynomial rings and their properties.

We will be rather formal with our notation in order to compare with the mathematical
literature. Consequently we will provide examples to map the formalism
to the analysis of the NHDM matrices $\mu_{ij}$ and $z_{ij,kl}$.

\subsection{The ring of invariants}

Let us consider a vector space $V$ as the space with dimension $\dim V$ spanned
by the basis $x_1,\dots,x_n$. The polynomial ring $K[V]=K[x_1,\dots,x_n]$ 
is then composed by polynomial functions in $x_i$ and span every algebraic combination
of the basis elements $x_i$ in the field $K$. We will consider the field
$K$ to be the complexes $\mathbb{C}$.
Furthermore, we consider the action of a group $G$ on $V$, for which each element $g \in G$ has
a representation $\rho(g)$ acting on $V$. We will abuse the notation
by stating $g$ instead of $\rho(g)$. 
Then we may define the ring of invariants $K[V]^G$ to be
\begin{equation}
    K[V]^G := \{ x \in K[V] \, \, | \, \, g.x = x \} \, ,
\end{equation}
which comprises all algebraic combinations of the parameters in $V$ which are
invariant under the action of $G$.
The ring of invariants $K[V]^G=K[f_1, \dots, f_r]$ is 
then generated by $f_1, \dots, f_r$, the primary invariants.

There are several noteworthy explanations so far. We begin with a space $V$ to which we apply the group $G$.
Then, we collect the invariants of the action of $G$ such that the remaining space is generated by elements 
called the
primary invariants. We already see that $K[V]^G \subseteq K[V]$, i.e. the invariant
ring is contained in the original one.
The dimension of a ring is called the Krull dimension and it is the minimum number
of generators of the ring.
The dimension of the initial ring $K[V]$
is given by 
\begin{equation}
    \dim (K[V]) = \dim V = \dim (K[x_1,\dots,x_n]) = n \, ,
\end{equation}
while for the ring of invariants we define the Krull
dimension as 
\begin{equation}
    \dim (K[V]^G) = \dim (K[f_1,\dots,f_r]^G) = r \, .
\end{equation}
Thus, $r \leq n$.
The Krull dimension of $K[V]^G$ also has a crucial
interpretation, it is the number of physical
parameters of the theory and it will be a meaningful quantity throughout this paper.

\begin{remark}
Let us consider a scalar potential with only $\mu_{ij}$. Then,
$V= \mathcal{A}$, the space of $2 \times 2$ Hermitian matrices
and $\mu \in \mathcal{A}$. The dimension of $V$ is then
$\dim V = 4$, the parameters $\mu_{11}, \mu_{12}, \mu^*_{12}, \mu_{22}$.
Hence, the polynomial
ring in the complexes is 
\begin{equation}
    \mathbb{C}[V] = \mathbb{C}[\mu_{11}, \mu_{12}, \mu^*_{12}, \mu_{22}]
    = \{ \mu_{11}, \mu_{12} + 2 \mu_{22}^2,  \dots \}
\end{equation}
consisting on all of the polynomials in the parameters.
If we choose to study the invariants under basis transformations $G=\mathrm{SU}(2)$
we may define the invariant ring as
\begin{equation}
    \mathbb{C}[V]^G = \mathbb{C}[\Tr \mu, \det \mu]^{\mathrm{SU}(2)}
    = \{ \Tr \mu, \det \mu + (\Tr \mu)^5,  \dots \} \ ,
\end{equation}
which consists on all polynomials which are $\mathrm{SU}(2)$-invariant.
The Krull dimension of $\mathbb{C}[V]^G$ is given by $\dim \mathbb{C}[V]^G = 2$. We
chose the two corresponding
parameters to be the trace and the determinant of $\mu_{ij}$.
\end{remark}

The ring of invariants $K[V]^G$ is solely responsible for all the physical information in
a theory. Knowing its generators is equivalent to knowing all of the physical
parameters, their CP properties, and in principle, even their impact on physical
processes.\footnote{The degrees of the generators may be used in principle
to track the order of a process, e.g. an invariant of degree six should not
appear in low order Feynman diagrams.}
To this end, we need a tool to
describe the generators of the invariant space, the Hilbert series.

For simplicity, while keeping some of the common mathematical
notation, we will interchange $K[V]$ with $R$
and $K[V]^G$ with $R^G$.

\subsection{The Hilbert series}

The Hilbert series is a very powerful tool for the characterization of $K[V]^G$, the ring of invariants.
The series itself is given by
\begin{equation}\label{eq:hil_series_deff}
    H(K[V]^G,t) = \sum_{k=0}^\infty \dim (R^G_k) \, t^k \, ,
\end{equation}
where $\dim (R^G_k)$ are the number of invariants of degree $k$ which are invariant under
the group $G$ and
$t$ is a token variable describing the degree of the invariants. The degree of the invariants
describes the degree of the polynomials in the parameters, e.g. in the previous
example $\Tr \mu$ is a degree
one invariant while $\det \mu$ is a degree two invariant.
Contrary to the primary invariants (generators), 
these need not be independent, as they total the number of
invariants.

The Hilbert series can also be written in a closed form, generally as
\begin{equation}\label{eq:general_hilbert_closed}
    H(K[V]^G,t) = \frac{P(t)}{Q(t)} = 
    \frac{P(t)}{ (1-t)^{d_1} (1-t^2)^{d_2} \dots (1-t^m)^{d_m}} \, ,
\end{equation}
where $P(t)$ is a polynomial. The denominator of the Hilbert series 
describes the degree and number of invariants under the action of a group $G$.
In eq.~\eqref{eq:general_hilbert_closed} we count $d_1$ invariants of degree one,
$d_2$ invariants of degree two, etc.

From eq.~\eqref{eq:general_hilbert_closed} we can read several properties
of the invariant ring $K[V]^G$.
In the context of the Hilbert series,
the Krull dimension $r$ is such that the limit
\begin{equation}\label{eq:lim_hil_krull}
    \lim_{t \rightarrow 1} (1-t)^r H(K[V]^G,t) = \gamma
\end{equation}
is neither infinite nor zero. Alternatively, this also means that in
eq.~\eqref{eq:general_hilbert_closed}
we can read the dimension as
\begin{equation}
    \dim (K[V]^G) = r = d_1 + d_2 + \dots + d_m \, .
\end{equation}
In other words, the minimum number of invariants needed to generate the invariant space
is given by the Krull dimension, $r$.
As such, we can always expand the Hilbert series as
\begin{equation}\label{eq:elem_expansion}
    H(K[V]^G,t) = \frac{\gamma}{(1-t)^r} + \frac{\tau}{(1-t)^{r-1}} + \dots \, ,
\end{equation}
where the significance of $\gamma$ and $\tau$ will be more clear later on.

\begin{remark}
In the 2HDM the Hilbert series has been fully computed \cite{Trautner:2018ipq},
both the ungraded (all token variables equal $t$) and the multi-graded (a token variable
for each representation). Excluding the three singlets of eq.~\eqref{eq:V_decomposition_2hdm},
the ungraded series is given by
\begin{equation}
    H(K[V]^G,t) = 1+4 t^2+4 t^3+15 t^4+18 t^5+53 t^6 + O\left(t^{7}\right) \, ,
\end{equation}
or in closed form
\begin{equation}\label{eq:hilbert_series_2hdm}
    H(K[V]^G,t) = \frac{1 + t^3 + 4t^4 + 2t^5 + 4t^6 + t^7 + t^{10}}
    {\left(1-t^2\right)^4 \left(1-t^3\right)^3 \left(1-t^4\right)} \, ,
\end{equation}
where we can read that $K[V]^G$ is generated by $4$ degree two, $3$ degree three and
$1$ degree four generators. Along with the three singlets, this yields a total of
$11$ physical parameters.
We can also expand it around $t=1$ such that
\begin{equation}
    H(K[V]^G,t) =  \frac{7/864}{(1-t)^{8}} + \frac{7/576}{(1-t)^{7}} + \dots \, .
\end{equation}
\end{remark}

\subsection{Molien series and the Weyl integration formula}

Until now we have just stated general properties of invariant rings and Hilbert series
but made no comment on its computation. For this, we separate two cases.

Let $G$ be a finite group and $\rho(g)$ a representation in $\mathrm{GL}_n$. Then
we may compute the Molien formula as
\begin{equation}\label{eq:molien_formula}
    H(K[V]^G,t) = \frac{1}{|G|} \sum_{g \in G} \frac{1}{\det (\mathds{1} - t \rho(g))} \, .
\end{equation}
We note that although we sum over the elements of the group, we need only to sum over
one element for each conjugacy class times the number of elements in it.

Similarly, one can compute the Hilbert series for an infinite group. Let $G$ be a reductive group,
e.g. $\mathrm{SU}(N)$, $\mathrm{SO}(N)$, $\mathrm{SL}(N)$. Then we define
the Weyl integration formula to be
\begin{equation}\label{eq:weyl_formula}
    H(K[V]^G,t) = \int_{G} \, d\mu_G  \frac{1}{\det (\mathds{1} - t \rho(g))} \, ,
\end{equation}
where $d\mu_G$ stand for the Haar measure. A number of them can be found in \cite{Hanany:2008sb}
for Lie groups, where it is defined as
\begin{equation}\label{eq:haar_measure}
    \int_{G} \, d\mu_G = \frac{1}{(2 \pi i)^m} \oint_{|z_1| = 1} \dots
    \oint_{|z_m| = 1} \frac{dz_1}{z_1}\dots \frac{dz_m}{z_m} \prod_{\alpha^+}
    \left( 1 - \prod_{l=1}^m z_l^{\alpha_l^+} \right) \, ,
\end{equation}
where $\alpha^+$ are the positive roots of the group.

Finally, we add the notion of plethystic exponential and plethystic logarithm, which as far
as we know was first introduced in \cite{getzler_kapranov_1998, labastida} and later for physical applications in \cite{Benvenuti:2006qr}.
The plethystic exponential
is defined as
\begin{equation}\label{eq:plethystic_exp}
    \mathrm{PE}[z_j,t,\mathbf{r}] := \exp \left( \sum_{k\geq 1} \frac{t^k \, \chi_\mathbf{r}(z_j^k )}{k}  \right) \, ,
\end{equation}
where $\chi_\mathbf{r}(z_j^k )$ is the character of the representation $\mathbf{r}=\rho$.
It can be interpreted with some trivial steps to be
\begin{equation}
    \frac{1}{\det (\mathds{1}-t\rho(g))}
    = \exp \left( \sum_{k \geq 1} \frac{t^k \Tr{\rho(g)^k} }{k} \right) \, .
\end{equation}
The plethystic logarithm is defined as
\begin{equation}\label{eq:plethystic_log}
    \mathrm{PL}\left[H(K[V]^G,t)\right] := \sum_{k \geq 1} \frac{\mu(k)}{k} \ln \left[ H(K[V]^G,t^k) \right] \, ,
\end{equation}
where $\mu(k)$ is the Möbius function. The significance of eq.~\eqref{eq:plethystic_log}
as a series is in the counting of possible primary invariants in the positive terms and
the determination of the syzygies in negative terms, and has been extensively discussed
in \cite{Benvenuti:2006qr}.

In the context of Lie groups we will always take the integration to be over the maximal
torus $\mathbb{T}$ of the group $G$. This will be the Abelian group which intersects all
conjugacy classes of $G$ and will greatly simplify our analysis.

\subsection{General properties}

Until now we have discussed several known results in invariant theory
as it applies to physics problems.
Here we
present a collection of formal results in invariant theory. These will be instrumental
to describe the class of models with more than one scalar. As with the more
formal sections of this text, we will give an example at the end to
guide the reader through the properties of the Hilbert series.

We have stated before that we will consider to be working
on $K= \mathbb{C}$, unless stated otherwise. 
We find that the characteristic of the field $\mathbb{C}$ is 
$\mathrm{char}(\mathbb{C})=0$.
In characteristic zero fields it suffices to say that for $K[V]^G$ to be finitely generated,
$G$ must be reductive. All semi-simple groups, finite groups and tori are examples, as described
in \cite{kemper}. Examples of semi-simple groups are $\mathrm{SL}(N)$, $\mathrm{SU}(N)$
and $\mathrm{O}(N)$.

\begin{theorem}\label{th:gorenstein}
If $G$ is semi-simple and connected, then $K[V]^G$ is Gorenstein \textup{\cite{Hochster}}. If
$K[V]^G$ is Gorenstein then
\begin{equation}\label{eq:gorenstein}
    H(K[V]^G,t^{-1}) = (-1)^r t^q H(K[V]^G,t) \, ,
\end{equation}
where $r$ is the Krull dimension and $q \in \mathbb{Z}$
as shown in \textup{\cite{Stanley1}}. 
\end{theorem}
In other words, eq.~\eqref{eq:gorenstein}
also implies that the numerator of the Hilbert series should be palindromic.

\begin{theorem}
A theorem in \textup{\cite{Popov1}} states that for almost all representations of a connected,
semi-simple group
$G$ we have 
\begin{equation}\label{eq:q_is_n}
q = \dim V \, ,
\end{equation}
where $q$ is defined in
eq.~\eqref{eq:gorenstein} and $\dim V$ is the dimension of our initial space.
\end{theorem}
\noindent
We will always assume
that this is true. The computation of the Hilbert series will confirm it at the end.
\begin{remark2}
The degree of the Hilbert series is defined by \textup{\cite{Popov1}}
\begin{equation}\label{eq:degree_hil_1}
    \deg \left( H(K[V]^G,t) \right) = \deg \left( \frac{P(t)}{Q(t)} \right)
    = \deg \left( P(t) \right) - \deg \left( Q(t) \right) = -q \, ,
\end{equation}
with $q$ defined in eq.~\eqref{eq:gorenstein}. Thus, it follows that for almost all
representations of $G$ we have that
\begin{equation}\label{eq:degree_hil_2}
    \deg \left( H(K[V]^G,t) \right) = -q = - \dim V \, .
\end{equation}
\end{remark2}
\noindent
This is an important result which will enable us to know
how to find the correct form of the Hilbert series
at the end of the computation. 
\begin{theorem}
A theorem of Knop and Littelmann \textup{\cite{Knop}} confirms that for all 
representations of $G$ we have
\begin{equation}
    r \leq - \deg \left( H(K[V]^G,t) \right) \leq \dim V \, .
\end{equation}
\end{theorem}

\noindent
In \cite{Popov1} another important corollary follows. 
\begin{corollary}
If $G$ is a semi-simple, connected
group $G$, then for almost all representations we have
\begin{equation}\label{eq:dimg_tau_gamma}
    \dim G = \frac{2 \tau}{\gamma} \, ,
\end{equation}
where $\gamma$ and $\tau$ are defined in eq.~\eqref{eq:elem_expansion}.
\end{corollary}
\noindent
The most important result we present here is given in \cite{Popov2}.
\begin{remark2}
If $G$ is semi-simple and
connected,
the Krull dimension $r$ is given by
\begin{equation}\label{eq:krull_almost}
    r = q - \frac{2 \tau}{\gamma} \longrightarrow \dim V - \dim G \, ,
\end{equation}
where the arrow means ``for almost all representations of $G$",
in accordance with eq.~\eqref{eq:q_is_n}.
\end{remark2}

Finally we can state a non-uniqueness property of the Hilbert series. Let $H_1(K[V]^G,t)$ be a Hilbert series
respecting the properties in 
eqs.~\eqref{eq:gorenstein}--\eqref{eq:degree_hil_2} and eqs.~\eqref{eq:dimg_tau_gamma}--\eqref{eq:krull_almost}.
Then there may exist $H_2(K[V]^G,t)$ with the same properties
such that
\begin{equation}
    H_1(K[V]^G,t) - H_2(K[V]^G,t) = 0 \, .
\end{equation}
A more general result
is discussed in \cite{kemper2} along with an algorithm to search for
an optimal solution which is often,
but not always,
the minimal solution. We will always search for the minimal solution, i.e. the one where the degrees
of the
Hilbert series are minimal.
An example may be provided in the 2HDM where
\begin{equation}\label{eq:nonuni_1}
    H_1(K[V]^G,t) = \frac{1 + t^3 + 4t^4 + 2t^5 + 4t^6 + t^7 + t^{10}}
    {\left(1-t^2\right)^4 \left(1-t^3\right)^3 \left(1-t^4\right)} \, ,
\end{equation}
and an alternative non-minimal Hilbert series,
\begin{equation}\label{eq:nonuni_2}
    H_2(K[V]^G,t) = \frac{1 + t^2 + t^3 + 4 t^4 + 3 t^5 + 8 t^6 + 3 t^7 + 
 4 t^8 + t^9 + t^{10} + t^{12}}
    {\left(1-t^2\right)^3 \left(1-t^3\right)^3 \left(1-t^4\right)^2} \, .
\end{equation}
We can readily check that the non-minimal solution uses only three degree two primary
invariants and two degree four primary invariants.
Thus, eqs.~\eqref{eq:nonuni_1}--\eqref{eq:nonuni_2} are an example
of the non-uniqueness of a Hilbert series.

\begin{remark}
Let $H(K[V]^G, t)$ be the Hilbert series of the 2HDM, written in eq.~\eqref{eq:nonuni_1}.
Theorem~\ref{th:gorenstein} states that its numerator is palindromic, which is true. 
While all of Popov's results are true for all but a finite number of representations, the 2HDM
is one of them, i.e. the representation $2(\mathbf{3}) \, \oplus \, \mathbf{5}$
is one of the
``almost all representations".
Then,
\begin{align}
    &q = 11 = \dim V \, , \\[2mm]
    &\deg \left( H(K[V]^G,t) \right) = - 11 = -\dim V \, ,  \\[2mm]
    &\dim G = \frac{2 \tau}{\gamma} = \frac{2\times 864}{576} = 3 \, , \\[2mm]
    & r = \dim K[V]^G = \dim V - \dim G = 11 - 3 = 8 \, .
\end{align}
With eq.~\eqref{eq:nonuni_2}, the exact same results can be extracted, although it
is not a minimal solution.
\end{remark}

\newpage

\section{Computing the Hilbert series}

While the computation of the Hilbert series by the Molien formula for finite groups enjoys
a large amount of software and information, the computation for infinite groups is much less
straightforward.

Calculating eq.~\eqref{eq:weyl_formula} is usually achieved by the use of the plethystic
exponential as the integrand using the characters of the representations of $G$. In
\cite{Lehman:2015via} there is a collection of such character functions in the appendix.
Ungraded and multi-graded Hilbert series are then computed with the residue
theorem and the fact that
\begin{equation}
    \mathrm{PE}[z_j,t,\mathbf{r}_1 \oplus \mathbf{r}_2] =
    \mathrm{PE}[z_j,t,\mathbf{r}_1] \times \mathrm{PE}[z_j,t,\mathbf{r}_2] \, ,
\end{equation}
as we know that $\Tr \left( a \oplus b \right) = \Tr(a) + \Tr(b)$.
Nevertheless, the use of the residue theorem has a stark impact on the complexity of the computation.
For large representations or multivariate integrations the closed form of the Hilbert
series may take too long to compute, too much memory or even be impossible with current
technology. This problem has since prevented the use of invariant theory in physics
for more complicated problems.

\subsection{Omega calculus}

Combinatorics has been a constant intersection with invariant theory.
In more recent years,
the same field has been essential for the computation of invariants.

In the work of Percy MacMahon \cite{Macmahon} the author illustrates partition analysis by solving combinatorics problems.
Suppose we want to find all non-negative integer solutions to $3a - 2b + c = 0$. Then, the generating function
will be an Elliott-rational function, a rational function which can be written as products in the denominator
of the type $A-B$, where $A$ and $B$ are monomials in the variables. Then it is characterized by
\begin{equation}
    \sum_{\substack{3a - 2b + c = 0 \\ a,b,c \geq 0}} t_1^a t_2^b t_3^c \, .
\end{equation}
Next, we introduce a new variable $\lambda$ and use an operator $\underset{=}{\Omega}$ to force the constant term
of the series, such that
\begin{equation}\label{eq:dio_1}
    \sum_{\substack{3a - 2b + c = 0 \\ a,b,c \geq 0}} t_1^a t_2^b t_3^c = 
    \underset{=}{\Omega} \sum_{a,b,c \geq 0} \lambda^{3a-2b+c} t_1^a t_2^b t_3^c \, .
\end{equation}
It can be shown that both the function before applying the operator and the one after are Elliott-rational functions
\cite{Macmahon}. 
In the $3a - 2b + c = 0$ example eq.~\eqref{eq:dio_1}
may be written in closed form as
\begin{equation}
    \frac{1 + t_1 \, t_2^2 \, t_3}{(1 - t_1^2 \, t_2^3)(1 - t_2 \, t_3^2)}
    =
    \underset{=}{\Omega} \, 
    \frac{1}{(1-\lambda^3 \, t_1)(1-\lambda^{-2} \, t_2)(1-\lambda \, t_3)} \, .
\end{equation}
The solution can be expanded in $1 + t_2 \, t_3^2 + t_1 \, t_2^2 \, t_3 +  \cdots$,
all of which are solutions to the Diophantine equation, e.g.
$3\times 0 - 2 \times 1 + 2 = 0$ and $3\times 1 - 2 \times 2 + 1 = 0$.

In general one defines the operator $\underset{=}{\Omega}$ as in \cite{Macmahon}
\begin{equation}\label{eq:omega_equal}
    \underset{=}{\Omega} \sum_{j_1=-\infty}^\infty  \dots \sum_{j_m=-\infty}^\infty a_{j_1, \dots , j_m} \lambda_1 \dots
    \lambda_m := a_{0, \dots , 0} \, ,
\end{equation}
where the variables $\lambda_i$ are restricted to the neighbourhood of $|\lambda_i|=1$. The computation of
such operation has been extensively covered by \cite{Andrews} which culminated with the development of the Omega
package for Mathematica and later, by Guoce Xin, for Maple \cite{Xin, XinP}. Besides the difference in platform,
Guoce Xin's software uses a faster algorithm based on a different approach detailed in his paper.

The fast algorithm in \cite{Xin} is a very powerful tool for computing Hilbert series with the Weyl formula.
It contrasts with the residue theorem as a faster and less resource hungry method and it is based on the following.
Let $G$ be a semi-simple group with a maximal torus $\mathbb{T}$ with an action on $V$ given by 
$\mathrm{diag} [m_1(z), \dots , m_n(z) ]$ where $m(z)$ are Laurent monomials
\cite{kemper} in $z_1, \dots, z_m$. Then
the character is given by 
\begin{equation}
    \chi_\mathbf{r} = \sum_{i=1}^n m_i(z) \, ,
\end{equation}
with $n=\dim V$.
Let us define with eq.~\eqref{eq:weyl_formula} and eq.~\eqref{eq:haar_measure} the Hilbert series
\begin{equation}
    H(K[V]^G,t) = \frac{1}{(2 \pi i)^m} \oint_{|z_1| = 1} \dots
    \oint_{|z_m| = 1} \frac{dz_1}{z_1}\dots \frac{dz_m}{z_m} 
    \frac{\prod_{\alpha^+}
    \left( 1 - \prod_{l=1}^m z_l^{\alpha_l^+} \right)}{(1-m_1(z) t) \dots (1-m_n(z) t)} \, .
\end{equation}
It then follows \cite{kemper} that the Hilbert series $H(K[V]^G,t)$ is the coefficient of $1$ as series
in $z_1, \dots, z_n$ of
\begin{equation}
    \frac{\prod_{\alpha^+}
    \left( 1 - \prod_{l=1}^m z_l^{\alpha_l^+} \right)}{(1-m_1(z) t) \dots (1-m_n(z) t)} \, .
\end{equation}
Thus, using Omega calculus, one can use eq.~\eqref{eq:omega_equal} to write the
important equality
\begin{equation}\label{eq:omega_hil}
    \boxed{H(K[V]^G,t) = \underset{=}{\Omega} \, \left[\frac{\prod_{\alpha^+}
    \left( 1 - \prod_{l=1}^m z_l^{\alpha_l^+} \right)}{(1-m_1(z) t) \dots (1-m_n(z) t)} \right]} 
    \, ,
\end{equation}
where instead of $\lambda_i$ we have $z_i$. We note that the assumption that $|z_i|=1$ is a straightforward assumption
in the Weyl formula. It is noteworthy that this algorithm will always work in a function of the type
$F(z_1,\dots, z_m) \in K[z_1,\dots, z_m,z_1^{-1},\dots, z_m^{-1}]$, i.e. it is described by powers of its variables and
reciprocals alone, which is a common property in many physical applications.

The speed of Xin's algorithm is owed to partial fraction decomposition and with it
we are able to remove entire rational functions which do not contribute to the final answer. This idea is
first attributed to Richard P. Stanley in \cite{Stanley2}.

\subsection{Pratical computation in Maple}\label{subsec:prac_maple}

The use of Xin's algorithm is fairly straightforward. It was used in \cite{Luque} with Ell.mpl 
shortly after its introduction and
in \cite{Xin2} where a version of eq.~\eqref{eq:omega_hil} is 
introduced and the package Ell2.mpl is used.
The package is rather straightforward with the main function
of our interest to be the command
\textit{E\_OeqW(f, v, ve)}, where \textit{f} is the integrand,
\textit{v} and \textit{ve} are all the variables and the variables
to integrate, respectively.

\begin{remark}
We will compute the Hilbert series for the 2HDM in \cite{Trautner:2018ipq}
starting with eq.~(5.3). Thus,
\begin{equation}
    H(K[V]^G,q,y,t) = \frac{1}{2 \pi i} \oint_{|z|=1} \frac{dz}{z} (1-z^2) \mathrm{PE}[z,q,\mathbf{5}]
    \mathrm{PE}[z,y,\mathbf{3}] \mathrm{PE}[z,t,\mathbf{3}] \, ,
\end{equation}
where 
\begin{align}
    &\mathrm{PE}[z,q,\mathbf{5}] \,
    \mathrm{PE}[z,y,\mathbf{3}] \, \mathrm{PE}[z,t,\mathbf{3}] = \nonumber \\[5mm]
    =&
    \frac{1}{(1-t)(1-\frac{t}{z^2})(1-t z^2)(1-y)(1-\frac{y}{z^2})(1-y z^2)
    (1-q)(1-\frac{q}{z^2})(1-q z^2)(1-\frac{q}{z^4})(1-q z^4)} \, ,
\end{align}
and $(1-z^2)=\prod_{\alpha^+} \left( 1 -  z^{\alpha^+} \right)$. Then, using
eq.~\eqref{eq:omega_hil} we have
\begin{equation}
    H(K[V]^G,t) = \underset{=}{\Omega} \, \left[
    (1-z^2)  \mathrm{PE}[z,q,\mathbf{5}] \,
    \mathrm{PE}[z,y,\mathbf{3}] \, \mathrm{PE}[z,t,\mathbf{3}] \right] \, ,
\end{equation}
where $z$ is the variable to eliminate and $q$, $y$ and $t$ are the
remaining variables. In Maple we write:

\begin{lstlisting}[language={},caption={}]
restart:

read("/path/to/Ell2.mpl")

integrand:=(q,y,t,z) -> write_our_eq.(4.11) * (1-z^2)

f := E_OeqW(integrand(q,y,t,z), [q,y,t,z], [z]):

g := normal(f)

\end{lstlisting}
\noindent
By simply multiplying the numerator
and denominator by $(1-q^2 y^2)(1-q^2 t^2)/[(1-q y)(1-q t)]$ it becomes clear
that we have successfully reproduced eq.~(5.4) of \cite{Trautner:2018ipq}.
\end{remark}

\begin{remark}
The ungraded Hilbert series in eq.~\eqref{eq:hilbert_series_2hdm}. While it is clear
that we may just do $q=y=t$ in our first example, we want to demonstrate the method when there are
powers in the denominator. Hence, we start with
\begin{align}
    \mathrm{PE}[z,t,\mathbf{5}] \,
    \mathrm{PE}[z,t,\mathbf{3}] \, \mathrm{PE}[z,t,\mathbf{3}] = 
    \frac{1}{(1-t)^3(1-\frac{t}{z^2})^3(1-t z^2)^3
    (1-\frac{t}{z^4})(1-t z^4)} \, ,
\end{align}
and thus
\begin{align}
    H(K[V]^G,t) &= \frac{1}{2 \pi i} \oint_{|z|=1} \frac{dz}{z} 
    \frac{(1-z^2)}{(1-t)^3(1-\frac{t}{z^2})^3(1-t z^2)^3
    (1-\frac{t}{z^4})(1-t z^4)} \nonumber \\[5mm]
    &= \underset{=}{\Omega} \, \left[
    \frac{(1-z^2)}{(1-t)^3(1-\frac{t}{z^2})^3(1-t z^2)^3
    (1-\frac{t}{z^4})(1-t z^4)} \right] \, .
\end{align}
Then we use Maple and write:
\begin{lstlisting}[language={},caption={}]
restart:

read("/path/to/Ell2.mpl")

integrand:=(t,z) -> (1-z^2)/((1-t)^3(1-t/z^2)^3(1-t*z^2)^3
    (1-t/z^4)(1-t*z^4))

f := E_OeqW(integrand(t,z), [t,z], [z]):

g := normal(f)

\end{lstlisting}
The output will be
\begin{equation}
    H(K[V]^G,t) = \frac{1-t^2 +t^3 +5 t^4 + t^5 -t^6 + t^8}
    {(1-t)^3(1-t^2)^3 (1-t^3)^2 (1+t)^2 (1+t+t^2) } \, ,
\end{equation}
which after some algebra we can write as
\begin{equation}
    H(K[V]^G,t) = \frac{1+t^3 +4t^4 +2 t^5 + 4t^6 + t^7 + t^{10}}
    {(1-t^2)^4(1-t^3)^3 (1-t^4) } \, ,
\end{equation}
in agreement with eq.~\eqref{eq:hilbert_series_2hdm}.
\end{remark}

\section{The 3HDM}

The case of the full characterization and counting of invariants in the 3HDM is still an open problem.
It is clear that the computation of their properties mirrors the one of the 2HDM, albeit the fact that
it is much more complicated. We will present here
for the first time the full computation of the Hilbert series of the 3HDM,
in both expanded and closed form.

\subsection{Definition of the Hilbert series}

Following our previous results in decomposing $V$ we may quickly
arrive at
the relevant decomposition of the 3HDM.
The decomposition of $z_{ij,kl}$ is given by
\begin{align}
    z_{ij,kl} &\rightarrow
    \left[ \mathrm{Sym}^2 (\mathbf{3}) \,
    \otimes \, \mathrm{Sym}^2 (\bar{\mathbf{3}} ) \right] \, \oplus \, 
    \left[ \mathrm{Alt}^2 (\mathbf{3}) \, \otimes \, \mathrm{Alt}^2 (\bar{\mathbf{3}} )  \right]
    \nonumber \\[2mm]
    &= \left[ \mathbf{6} \, \otimes \, \bar{\mathbf{6}} \right]
    \, \oplus \, \left[ \bar{\mathbf{3}} \, \otimes \, \mathbf{3} \right] \nonumber \\[2mm]
    &= \left[ \mathbf{1} \, \oplus \, \mathbf{8} \, \oplus \, \mathbf{27} \right]
    \, \oplus \, \left[ \mathbf{1} \, \oplus \, \mathbf{8} \right] \nonumber \\[2mm]
    &=  \mathbf{1} \, \oplus \, \mathbf{1} \, \oplus \, \mathbf{8} \oplus \, \mathbf{8} 
    \oplus \, \mathbf{27} \, .
\end{align}
The decomposition of $\mu_{ij}$ is straightforward and hence,
\begin{align}\label{eq:decomp_3hdm}
    & \mu_{ij} \rightarrow \mathbf{1} \, \oplus \, \mathbf{8} \, , \nonumber \\[5mm]
    & z_{ij,kl} \rightarrow \mathbf{1} \, \oplus \, \mathbf{1} \, \oplus \, \mathbf{8}
    \oplus \, \mathbf{8} 
    \oplus \, \mathbf{27} \, ,
\end{align}
a result that we compute differently in appendix~\ref{sec:appA}.
From here, we can already define the Hilbert series as in
eq.~\eqref{eq:weyl_formula}
with the plethystic exponentials. Thus,
the multigraded Hilbert series we are interested in is defined as
\begin{align}
    H(K[V]^G,s,t,u,q) = & \frac{1}{(2 \pi i)^2} \oint_{|z_1|=1} \frac{d z_1}{z_1}
    \oint_{|z_2|=1} \frac{d z_2}{z_2} (1-z_1 z_2)\left( 1 - \frac{z_1^2}{z_2} \right)
    \left( 1 - \frac{z_2^2}{z_1} \right) \times \nonumber \\[5mm]
    & \mathrm{PE}[z_1,z_2,s,\mathbf{8}] \, \mathrm{PE}[z_1,z_2,t,\mathbf{8}] \,
     \mathrm{PE}[z_1,z_2,u,\mathbf{8}] \, \mathrm{PE}[z_1,z_2,q,\mathbf{27}] \, ,
\end{align}
where the token variables are $s$, $t$, $u$ for the
three adjoint representations and $q$ for the $\mathbf{27}$. 
From eq.~\eqref{eq:plethystic_exp} we can compute the plethystic exponentials of the $\mathbf{8}$'s
and the $\mathbf{27}$.
The plethystic exponential depends only on the character polynomials that we construct with the 
weight system of the irreducible representations. With LieART \cite{Feger:2019tvk} it is straightforward
to compute
\begin{align}
    \chi_{\mathbf{8}}(z_1,z_2) = z_1 z_2 + \frac{z_2^2}{z_1} + \frac{z_1^2}{z_2} + 2 +
    \frac{z_2}{z_1^2} + \frac{z_1}{z_2^2} + \frac{1}{z_1 z_2} \, ,
\end{align}
and
\begin{align}
    \chi_\mathbf{27}(z_1,z_2) = & \frac{z_1^4}{z_2^2} + \frac{z_2^2}{z_1^4} + \frac{z_1^3}{z_2^3}
    + \frac{z_2^3}{z_1^3} + z_1^3 + \frac{1}{z_1^3} + \frac{z_1^2}{z_2^4} + \frac{z_2^4}{z_1^2}
    + z_1^2 z_2^2 + \frac{1}{z_1^2 z_2^2} \nonumber \\[5mm]
    &+ \frac{2z_1^2}{z_2}  + \frac{2z_2}{z_1^2} + \frac{2z_1}{z_2^2}
    + \frac{2z_2^2}{z_1} + 2 z_1 z_2 + \frac{2}{z_1 z_2} + z_2^3 + \frac{1}{z_2^3} + 3 \, .
\end{align}
Then, through eq.~\eqref{eq:plethystic_exp} we have
\begin{align}\label{eq:ple8}
    \mathrm{PE}[z_1,z_2,s,\mathbf{8}] = & \left[
    (1-s)^2 \left(1-s \frac{z_1^2}{z_2}\right) \left(1-s \frac{z_2}{z_1^2}\right)
    \left(1-s \frac{1}{z_1 z_2}\right) \times \right.
    \nonumber \\[2mm]
    & \left. \left(1-s z_1 z_2\right) \left(1-s \frac{z_1}{z_2^2}\right) \left(1-s \frac{z_2^2}{z_1}\right) \right]^{-1}
    \, ,
\end{align}
and
\begin{align}\label{eq:ple27}
    \mathrm{PE}[z_1,z_2,q,\mathbf{27}] = & \left[
    (1-q)^3 \left(1-q \frac{1}{z_2^3} \right) \left(1-q z_2^3\right) \left(1-q \frac{1}{z_1 z_2} \right)^2 
    \left(1-q z_1 z_2 \right)^2 \left(1-q \frac{z_2^2}{z_1} \right)^2 \times \right. \nonumber \\[2mm]
    &  \left(1-q \frac{z_1}{z_2^2}\right)^2 \left(1-q \frac{z_2}{z_1^2} \right)^2 \left(1- q \frac{z_1^2}{z_2} \right)^2 
    \left(1- q \frac{1}{z_1^2 z_2^2} \right) \left(1 - q z_1^2 z_2^2 \right)  \times \nonumber \\[2mm]
    &  \left(1-q \frac{z_2^4}{z_1^2} \right) \left(1-q \frac{z_1^2}{z_2^4} \right) 
    \left(1-q \frac{1}{z_1^3} \right) \left(1-q z_1^3\right)
    \left(1-q \frac{z_2^3}{z_1^3} \right)  \times \nonumber \\[2mm]
    & \left.  \left(1-q \frac{z_1^3}{z_2^3} \right) \left(1-q \frac{z_2^2}{z_1^4} \right) 
    \left(1-q \frac{z_1^4}{z_2^2} \right) \right]^{-1} \, ,
\end{align}
which we already recognize as Elliott-rational functions.
It may be inferred by examining eqs.~\eqref{eq:ple8} and \eqref{eq:ple27} that
the computation of this particular integral is difficult due to the existence
of higher order poles, cubic and quartic polynomials as well as the plain fact that we are dealing with
multivariate residues. These are known to be specially difficult to handle.

\subsection{Expansions and plethystic logarithm}

An alternative to the direct computation of the Hilbert series is an expansion as a formal
series in the token variables. As it turns out this expansion is
well-behaved and easy to compute. With it we may apply the residue theorem
for $z_1, \, z_2 \rightarrow 0$ after truncating the series. Then,
\begin{align}
    H(K[V]^G,s,t,u,q) = & 1+q^2+u^2+t u+t^2+s u+s t+s^2+2 q^3+q^2 u+q u^2+u^3+q^2 t \nonumber \\
    & +q t u+t u^2+q t^2+t^2 u+t^3+q^2 s+q s u+s u^2+q s t+2 s t u+s t^2 \nonumber \\
    & +q s^2+s^2 u+s^2 t+s^3+4 q^4+2 q^3 u+4 q^2 u^2+q u^3+u^4+2 q^3 t+5 q^2 t u \nonumber \\
    & +3 q t u^2+t u^3+4 q^2 t^2+3 q t^2 u+3 t^2 u^2+q t^3+t^3 u+t^4+2 q^3 s+5 q^2 s u \nonumber \\
    & +3 q s u^2+s u^3+5 q^2 s t+6 q s t u+4 s t u^2+3 q s t^2+4 s t^2 u+s t^3+4 q^2 s^2 \nonumber \\
    & +3 q s^2 u+3 s^2 u^2+3 q s^2 t+4 s^2 t u+3 s^2 t^2+q s^3+s^3 u+s^3 t+s^4+6 q^5 \nonumber \\
    & +8 q^4 u+11 q^3 u^2+5 q^2 u^3+2 q u^4+u^5+8 q^4 t+17 q^3 t u+14 q^2 t u^2+6 q t u^3 \nonumber \\
    & +2 t u^4+11 q^3 t^2+14 q^2 t^2 u+10 q t^2 u^2+3 t^2 u^3+5 q^2 t^3+6 q t^3 u+3 t^3 u^2+2 q t^4 \nonumber \\
    & +2 t^4 u+t^5+8 q^4 s+17 q^3 s u+14 q^2 s u^2+6 q s u^3+2 s u^4+17 q^3 s t+27 q^2 s t u \nonumber \\
    & +17 q s t u^2+6 s t u^3+14 q^2 s t^2+17 q s t^2 u+8 s t^2 u^2+6 q s t^3+6 s t^3 u+2 s t^4 \nonumber \\
    & +11 q^3 s^2+14 q^2 s^2 u+10 q s^2 u^2+3 s^2 u^3+14 q^2 s^2 t+17 q s^2 t u+8 s^2 t u^2 \nonumber \\
    & +10 q s^2 t^2+8 s^2 t^2 u+3 s^2 t^3+5 q^2 s^3+6 q s^3 u+3 s^3 u^2+6 q s^3 t+6 s^3 t u \nonumber \\
    & +3 s^3 t^2+2 q s^4+2 s^4 u+2 s^4 t+s^5 + \mathcal{O}\left( \left[ stuq \right]^6 \right) \, .
\end{align}
The ungraded Hilbert series is then given by equaling $t=s=u=q$,
\begin{align}
    H(K[V]^G,t) =& 1 + 7 \, t^2 + 22 \, t^3 + 94 \, t^4 + 438 \, t^5 + 1971 \, t^6 + 8376 \, t^7 
    + 34973 \, t^8 + 138426 \, t^9 \nonumber \\[2mm] 
    &+ 525486 \, t^{10} + 1912602 \, t^{11} 
    + 6685563 \, t^{12}+ 22488737 \, t^{13} + 72974065 \, t^{14} \nonumber \\[2mm] 
    &+ 228829031 \, t^{15} + 694812413 \, t^{16} + 2046440237 \, t^{17} + 5856320772 \, t^{18} \nonumber \\[2mm] 
    &+ 16308266932 \, t^{19} + 44255437022 \, t^{20} + \mathcal{O}\left( t^{21} \right) \, ,
\end{align}
where we see a direct interpretation with eq.~\eqref{eq:hil_series_deff}. It is important
to note that these invariants are not necessarily algebraically independent, as this distinction
will be computed only with the closed form of the Hilbert series.

The plethystic logarithm will be given by eq.~\eqref{eq:plethystic_log} and it will allow us
to know the type of invariants relevant for our study. The expansion is very long and we present
the degree two and three invariants along with the first syzygy, i.e. 
the first negative term. The expansion is
\begin{align}
    \mathrm{PL}[H(K[V]^G,s,t,u,q)] =& \,
    q^2 + u^2 + t u + t^2 + s u + s t + s^2 + 2 q^3 + q^2 u + q u^2 + u^3   \nonumber \\[2mm]
    &+ q^2 t + 
 q t u + t u^2 + q t^2 + t^2 u + t^3 + q^2 s + q s u + s u^2 + q s t  \nonumber \\[2mm]
 &+ 2 s t u +
  s t^2 + q s^2 + s^2 u + s^2 t + s^3 + \dots  \nonumber \\[2mm]
  &- s^2 t^2 u^3 + \dots \, .
\end{align}
From here, we are only missing the information from the closed form of Hilbert series, how many
invariants and of what degree.

\subsection{The Hilbert series of the 3HDM}

The Hilbert series of the 3HDM would naively be computed with the use of 
the residue theorem. By doing so, one quickly finds the computation to be
very complex as various problems come into play. First, by solving
higher degree polynomials in the denominator and integrating the first time we
arrive at a second integration plagued with square, cubic and quartic roots
of the integration variable. Secondly, it is not trivial to deal with the complex
roots nor to use substitution of variables in the integrand.
Omega calculus, used here for the first time for NHDM physical applications,
offers the solution for all of these shortcomings.

In the 3HDM we have
\begin{equation}
    \prod_{\alpha^+}
    \left( 1 - \prod_{l=1}^m z_l^{\alpha_l^+} \right) = 
    \left( 1 - z_1 z_2 \right) \left( 1 - \frac{z_1^2}{z_2} \right) 
    \left( 1 - \frac{z_2^2}{z_1} \right)
\end{equation}
and then by eq.~\eqref{eq:omega_hil} and eq.~\eqref{eq:decomp_3hdm} 
we write the ungraded Hilbert series
as
\begin{equation}\label{eq:omega_hil_3hdm}
    H(K[V]^G,t) = \underset{=}{\Omega} \, \left[
    \left( 1 - z_1 z_2 \right) \left( 1 - \frac{z_1^2}{z_2} \right) 
    \left( 1 - \frac{z_2^2}{z_1} \right)  \mathrm{PE}[z,t,\mathbf{8}]^3 \,
    \mathrm{PE}[z,t,\mathbf{27}] \right] \, .
\end{equation}
The code in Maple is straightforward and described in subsection~\ref{subsec:prac_maple}.
It ran for $56$ minutes using $5 \, \mathrm{Gb}$ of memory in a personal laptop equipped
with an Intel Core i7-8750H. The solution, while not immediately in the form most useful to us,
consists on the rational function
\begin{align}\label{eq:hil_3hdm_first}
    H(K[V]^G,t) = & \frac{P_{146}(t)}
    {(1 - t)^{43} (1 + t)^{20} (1 + t^2)^{10} (1 + t + t^2)^{16} (1 + t +
    t^2 + t^3 + t^4)^9 } \nonumber \\[5mm]
   &\times
   \frac{1}{(1 + t^3 + t^6) (1 + t^2 + t^3 + t^4 + t^5 + t^6 + t^8)^5} \, ,
\end{align}
where we refrained from writing the full palindromic polynomial of degree $146$ in the
numerator. 

Before going forward we note several interesting properties of eq.~\eqref{eq:hil_3hdm_first}.
First, the Krull dimension is $43$, or equivalently the number of physical parameters minus the 
three singlets. This
comes directly from eq.~\eqref{eq:krull_almost} as $(27+3\times8) - 8 = 43$. We will
make this connection exact in the next section. Second, by expanding around $t = 1$ we get
\begin{equation}
    H(K[V]^G,t) = 
    \frac{\lambda}{(1-t)^{43}}
    +  \frac{\tau}{(1-t)^{42}} 
    + \mathcal{O}((1-t)^{-41}) \, ,
\end{equation}
in agreement with eq.~\eqref{eq:elem_expansion}, and
\begin{align}
    \lambda &= \frac{1936873185248320313}{33716902552613683200000000} \, , \nonumber \\[5mm]
    \tau &= \frac{1936873185248320313}{8429225638153420800000000} \, .
\end{align}
From eq.~\eqref{eq:dimg_tau_gamma} we also add validity to our earlier assumptions
that eq.~\eqref{eq:krull_almost} is valid in the 3HDM. Thus,
\begin{equation}
    \frac{2 \tau}{\gamma} = 8 = \dim G \, ,
\end{equation}
achieving the expected result.

From eq.~\eqref{eq:hil_3hdm_first} it is not trivial to find a minimal Hilbert
series that satisfies the same Krull dimension and also
eqs.~\eqref{eq:gorenstein}--\eqref{eq:degree_hil_2} and
eqs.~\eqref{eq:dimg_tau_gamma}--\eqref{eq:krull_almost}. Simple algebraic manipulations
lead us to thousands of solutions. 
Thus, we devise a brute-force algorithm. We start by expanding eq.~\eqref{eq:hil_3hdm_first}
to $500$ terms and then we multiply it by various possible denominators with Krull dimension
$43$. Then,
we filter the numerator and test it for the palindromic property while requiring the coefficients
to be non-negative.
After an intensive search using
the NumPy package in Python we get to a seemingly minimal solution.\footnote{NumPy
turns out to be much faster with numpy.poly1d() at multiplying and manipulating
polynomials than SymPy.}
Thus, the Hilbert series describing the most general 3HDM is given by
\footnote{In the spirit of appendix~\ref{sec:appB} this is the $\mathrm{SU}(3)$
Hilbert series of three $\mathbf{8}$'s and one $\mathbf{27}$.}
\begin{equation}\label{eq:hil_3hdm_sec}
    \boxed{ H(K[V]^G,t) = \frac{P_{166}(t)}
    {\left(1-t^2\right)^7 \left(1-t^3\right)^8 \left(1-t^4\right)^6 \left(1-t^5\right)^9
    \left(1-t^6\right)^3 \left(1-t^7\right)^5 \left(1-t^9\right) \left(1-t^{12}\right)^4} }
\end{equation}
where the palindromic polynomial
$P_{166}(t)$ is too large to write here but we write it 
in subappendix~\ref{app:B_3hdm} and in an ancillary file attached to this
paper. We already see that eq.~\eqref{eq:hil_3hdm_sec} also agrees with
eq.~\eqref{eq:degree_hil_1} and eq.~\eqref{eq:degree_hil_2} in that
\begin{equation}
    \deg \left( H(K[V]^G,t) \right) = -51 = - \dim V = -q \, .
\end{equation}
Furthermore, we note the large degree invariants in eq.~\eqref{eq:hil_3hdm_sec}
contrasting with the case of the 2HDM.
The question remains if this Hilbert series is minimal. Although
we are confident with the result,
only a subsequent study on the invariants themselves can point to whether
this is an optimal solution. This will be the topic of a future paper \cite{to_come}.

\section{Properties of the NHDM}

In this section we will work out a number of interesting
properties of the NHDM, which we can learn from the tools
used so far.
Our first result concerns the counting of the physical parameters of the NDHM.
\begin{theorem}\label{th:1}
Let the model be the most general NHDM. Then, the number of physical parameters is given by
\begin{equation}
    N_\mathrm{physical} = \frac{N^4+N^2+2}{2} \, ,
\end{equation}
where $N$ are the number of doublets.\footnote{This result was first conjectured and stated by
J. P. Silva in a private discussion built on table 1 of \cite{Ferreira:2008zy}. 
Here, we show a formal proof of it.}
\end{theorem}

\begin{proof}
Let the group $G=\mathrm{SU}(N)$ be the family transformations of the NHDM and let us define
a physical parameter as a family invariant parameter. Then we may define the invariant ring
$K[V]^G$ as having Krull dimension $N_\mathrm{physical}$ and field $K=\mathbb{C}$.
The dimension of the initial space $\dim V$ is then given by
\begin{equation}
    \dim V = \sum_i \dim \mathbf{r}_i = \frac{N^2(N^2+3)}{2} \, ,
\end{equation}
where $\mathbf{r}_i$ are the representations of the decomposition of the matrices $\mu$ and $z$,
and the last equality is given by the counting of total parameters.
Then, we use a theorem in \cite{Ros}, asserting that the Krull dimension of $K[V]^G$ is given by
\begin{equation}
    \dim K[V]^G = \dim V - \dim G +  \dim G_v  \, ,
\end{equation}
where $G_v$ is the stabilizer of $G$.\footnote{In fact there is a distinction here between the notion
of transcendence degree and Krull dimension. Nevertheless, we don't need to worry about it as
they are the same in finitely generated algebras.} Because we already know that $\dim G = N^2-1$,
we only need to compute the dimension of the stabilizer for our case. Specifically, whether
is zero 
($G$ acts freely on $V$) or not.
In ref.~\cite{Popov3}
the authors establish that if the action is reducible, and it is in our case, then $G$ acts freely
if at least one irreducible action acts freely. In particular, for the irreducible
representation of a simple group, which is the case
of $\mathrm{SU}(N)$, $\dim G_v = 0$ if and only if $\dim \mathbf{r}_i > \dim G$.
Thus, we only have to show that in the NHDM, there is always an irreducible representation
with dimension greater than $N^2-1$. This is trivial because the decomposition of the NHDM
always implies the computation of $\mathbf{r}_a \, \otimes \, \mathbf{r}_a$ for $\mathbf{r}_a$
being the adjoint representation. This will always result in at least a representation
of higher dimension than $\dim G$ which will always be needed for the decomposition. Hence,
in the NHDM
\begin{align}
    \dim K[V]^G &= \dim V - \dim G +  \dim G_v \nonumber \\[2mm]
    &= \frac{N^2(N^2+3)}{2} - (N^2-1) + 0 \nonumber \\[2mm]
    &= \frac{N^4+N^2+2}{2} \, .
\end{align}
\end{proof}
Our proof sheds light on the conditions of
eq.~\eqref{eq:krull_almost} and shows its validity for a number of cases. This result does
not hold in general for cases where symmetries are enforced in the Lagrangian. It does however
bound the number of physical parameters in any NHDM.
We summarize theorem~\ref{th:1} in table~\ref{table:nhdm_physical}.
\begin{table}[ht]
\begin{center}
\begin{tabular}{|c|c|c|}
\hline
N & Number of parameters ($\dim V$) & $N_\mathrm{physical} = \dim K[V]^G$ \\
\hline
\hline
$2$ & $14$ & $11$ \\
\hline
$3$ & $54$ &
$46$ \\
\hline
$4$ & $152$ & 
$137$ \\
\hline
$5$ & $350$ &
$326$ \\
\hline
$6$ & $702$ &
$667$ \\
\hline
$7$ & $1274$ &
$1226$ \\
\hline
\dots & \dots &
\dots \\
\hline
$N$ & $\frac{N^2(N^2+3)}{2}$ & $\frac{N^4+N^2+2}{2}$ \\
\hline
\end{tabular}
\caption{\label{table:nhdm_physical} Physical parameters of the NHDM with a group
of family transformations $\mathrm{SU}(N)$.}
\end{center}
\end{table}

Another result we provide regards the decomposition
of multi-Higgs doublet models.
\begin{theorem}\label{th:rf}
Let the model be the NHDM with $N>3$. Then
the vector space of parameters $V$ is decomposed as
\begin{align}\label{eq:V_decomp}
    V = 3(\mathbf{1}) \, \oplus \, 3\left( \mathbf{N^2-1} \right)  \, 
    \oplus \, \left( \mathbf{\frac{N^2(N+1)(N-3)}{4}} \right) 
    \, \oplus \, \left( \mathbf{\frac{N^2(N-1)(N+3)}{4}} \right) \, ,
\end{align}
with
\begin{align}
    \dim V &= 3 \, + \, 3\left( N^2-1 \right)  \, 
    + \, \left( \frac{N^2(N+1)(N-3)}{4} \right) 
    \, + \, \left( \frac{N^2(N-1)(N+3)}{4} \right) \nonumber \\[2mm]
    &= \frac{N^2(N^2+3)}{2} \, .
\end{align}
The decomposition of $\mu_{ij}$ and $z_{ij,kl}$ is given by
\begin{align}\label{eq:gen_rep_decom}
    \mu_{ij} &\rightarrow \mathbf{1} \, \oplus \, (\mathbf{N^2-1}) \, , \nonumber \\[5mm]
    z_{ij,kl} &\rightarrow 2(\mathbf{1}) \, \oplus \, 2\left( \mathbf{N^2-1} \right)  \, 
    \oplus \, \left( \mathbf{\frac{N^2(N+1)(N-3)}{4}} \right) 
    \, \oplus \, \left( \mathbf{\frac{N^2(N-1)(N+3)}{4}} \right) \, .
\end{align}
\end{theorem}
The proof for this theorem is given in appendix~\ref{sec:appA}.\footnote{We thank Renato Fonseca
for providing the outline of this proof in a private communication.}
\begin{table}[t]
\begin{center}
\begin{tabular}{|c|c|c|c|}
\hline
N & $\mu_{ij}$ & $z_{ij,kl}$ & Number of parameters ($\dim V$) \\
\hline
\hline
$2$ & $\mathbf{1} \, \oplus \, \mathbf{3}$ & $ 2(\mathbf{1}) \, \oplus \, \mathbf{3} \,
\oplus \, \mathbf{5} $ & $14$ \\
\hline
$3$ & $\mathbf{1} \, \oplus \, \mathbf{8}$ &
$2(\mathbf{1}) \, \oplus \, 2(\mathbf{8}) \, \oplus \, \mathbf{27}$ & $54$ \\
\hline
$4$ & $\mathbf{1} \, \oplus \, \mathbf{15}$ & 
$2(\mathbf{1}) \, \oplus \, 2(\mathbf{15}) \, \oplus \, \mathbf{20} \, \oplus \, \mathbf{84}$ & $152$ \\
\hline
$5$ & $\mathbf{1} \, \oplus \, \mathbf{24}$ &
$2(\mathbf{1}) \, \oplus \, 2(\mathbf{24}) \, \oplus \, \mathbf{75} \, \oplus \, \mathbf{200}$ & $350$ \\
\hline
$6$ & $\mathbf{1} \, \oplus \, \mathbf{35}$ &
$2(\mathbf{1}) \, \oplus \, 2(\mathbf{35}) \, \oplus \, \mathbf{189} \, \oplus \, \mathbf{405}$ & $702$ \\
\hline
$7$ & $\mathbf{1} \, \oplus \, \mathbf{48}$ &
$2(\mathbf{1}) \, \oplus \, 2(\mathbf{48}) \, \oplus \, \mathbf{392} \, \oplus \, \mathbf{735}$ & $1274$ \\
\hline
\dots & \dots &
\dots & \dots \\
\hline
$N$ & $\mathbf{1} \, \oplus \, (\mathbf{N^2-1})$ & $2(\mathbf{1}) \, \oplus \, 2\left( \mathbf{N^2-1} \right)  \, 
    \oplus \, \mathbf{a}_N
    \, \oplus \, \mathbf{b}_N $
& $\frac{N^2(N^2+3)}{2}$ \\
\hline
\end{tabular}
\caption{\label{table:nhdm_reps} Representation decomposition of the NHDM where
$a_N$ and $b_N$ are given in eq.~\eqref{eq:bn_dn}.}
\end{center}
\end{table}
We summarize our results in table~\ref{table:nhdm_reps}, where we
used the LieART package for
Mathematica to confirm results. The functions $a_N$ and $b_N$ are defined as
\begin{align}\label{eq:bn_dn}
    a_N = \frac{N^2(N+1)(N-3)}{4} \quad \mathrm{and} \quad b_N = \frac{N^2(N-1)(N+3)}{4} \, ,
\end{align}
and the last line is representative of theorem~\ref{th:rf}.

\section{Parameter counting with symmetries}\label{sec:symmetries}

So far we have discussed both the decomposition of the matrices
of the Lagrangian, and the invariants of the most
general multi-Higgs scalar models. However, we have not presented any
result towards the use of symmetries in the Lagrangian. In this section
we show how to count all of the remaining parameters after imposing a symmetry.

The main idea of this technique is to enumerate invariants
for a symmetry group $G$. By doing so, we count the number of parameters
of the Lagrangian, which in specific cases might be larger than
the number of independent physical parameters.
In column two of table~\ref{table:nhdm_physical}
we obtained the number of physical parameters needed
to describe a generic NHDM.
Using the basis freedom, we could reduce the number of parameters
to that obtained in column three of table~\ref{table:nhdm_physical}.
The numbers obtained in this section parallel those obtained in
column two of table~\ref{table:nhdm_physical},
but now for a symmetry-constrained NHDM.
Indeed, specific groups might still allow for some remnant basis
freedom, which might be used to reduce the number of independent
parameters required.
Below, an example is provided in the 2HDM with $\mathbb{Z}_2$
symmetry and subsequent discussion after eq.~\eqref{eq:2hdm_z2_z}.
As we will see, applying $\mathbb{Z}_2$ to the 2HDM
still allows for a rephasing of the second doublet;
a freedom which may be used in order to cancel the imaginary
part of one quartic coupling.
In such examples, one may reduce the counting of
parameters by choosing a specific basis. Nevertheless, the maximum
number of parameters that remain in the Lagrangian are given
by the following basis-invariant technique.

\begin{theorem}
Let us consider a symmetry by the action of a group $G$. We choose a representation
$\rho(g) = \mathbf{r} = \bigoplus \mathbf{r}_i$ for the fields. 
Then the number of parameters is given by
the number of singlets in
\begin{align}
    &\mu_{i j} \rightarrow \bar{\mathbf{r}}  \, \otimes \, \mathbf{r} \, , \nonumber \\[2mm]
    &z_{ij,kl} \rightarrow \left[ \mathrm{Sym}^2 (\mathbf{r}) \,
    \otimes \, \mathrm{Sym}^2 (\bar{\mathbf{r}} ) \right] \, \oplus \, 
    \left[ \mathrm{Alt}^2 (\mathbf{r}) \, \otimes \, \mathrm{Alt}^2 (\bar{\mathbf{r}} )  \right] \, .
\end{align}
\end{theorem}
\begin{proof}
By imposing a symmetry in the Lagrangian we are decomposing the vector space
$V = \mu \, \oplus \, z$ in irreducible representations of $G$. In contrast with
the strategy for basis invariants, we know that only the degree one invariants can remain.
These correspond to elements $(1-t)$ in the denominator of the Hilbert series.
A remarkable property of these terms is that they can be factored out from
the summation as they do not depend on the representation.\footnote{In fact
we had already silently agreed to this when we left out the three singlets
of the 3HDM from the computation of the Hilbert series.}
This can be understood by
decomposing
\begin{align}\label{eq:rt}
    \mathbf{r}_T &= (\bar{\mathbf{r}}  \, \otimes \, \mathbf{r}) \, \oplus \,
    \left[ \mathrm{Sym}^2 (\mathbf{r}) \,
    \otimes \, \mathrm{Sym}^2 (\bar{\mathbf{r}} ) \right] \, \oplus \, 
    \left[ \mathrm{Alt}^2 (\mathbf{r}) \, \otimes \, \mathrm{Alt}^2 (\bar{\mathbf{r}} )  \right]
    \nonumber \\[2mm]
    &= n(\mathbf{1})  \, \oplus \, \bigoplus_{\mathbf{r}_j \neq \mathbf{1}} \mathbf{r}_j
    \, ,
\end{align}
where $\mathbf{r}_T$ is the representation of the full decomposition.
Then,
\begin{equation}
    H(K[V]^G, t) = \SumInt \frac{1}{\det \left( \mathds{1} - t \, \mathbf{r}_T \right) }
    =  \frac{1}{(1-t)^n}
    \SumInt
    \frac{1}{\det \left( \mathds{1} - t \, \bigoplus_{\mathbf{r}_j 
    \neq \mathbf{1}} \mathbf{r}_j  \right)}
    \, ,
\end{equation}
and thus,
\begin{equation}
    H(K[V]^G, t) = (1-t)^{-n} H(K[V']^G, t) \, .
\end{equation}
The remaining irreducible representations $\mathbf{r}_j$ are not invariant by themselves
and require higher degrees to form an invariant.
Consequently, we need not compute the Hilbert series to know how many
invariants of degree one exist. This quantity is given by $n$ in eq.~\eqref{eq:rt},
the number of singlets.
\end{proof}
\begin{remark}
We consider the 2HDM with a $\mathbb{Z}_2$ symmetry. We choose the representation
$\mathbf{r}=\mathbf{1} \, \oplus \, \mathbf{1}'$ corresponding to the transformation
$\mathrm{diag} (1,-1)$. Then, with the product rule
\begin{align}
    \mathbf{1}' \, \otimes \, \mathbf{1}' = \mathbf{1} \, ,
\end{align}
we get
\begin{equation}
    \mu_{ij} \rightarrow (\mathbf{1} \, \oplus \, \mathbf{1}')^{\otimes 2}
    = 2(\mathbf{1}) \, \oplus \, 2(\mathbf{1}') \, .
\end{equation}
For $z_{ij,kl}$ we need first to know what corresponds to $\mathrm{Sym}$ and $\mathrm{Alt}$.
This can easily be achieved with the character table and with
\begin{align}
    \chi_{\mathrm{Sym}^2}(g) &= \frac{1}{2} \left( \chi(g)^2 + \chi(g^2) \right) \, ,
    \nonumber \\[2mm]
    \chi_{\mathrm{Alt}^2}(g) &= \frac{1}{2} \left( \chi(g)^2 - \chi(g^2) \right) \, ,
\end{align}
where $\chi(g)$ is the character of $g$. This system always has a solution. Here, by choosing
$g=a$, with $a^2 = e$, we have
\begin{align}
    \chi_{\mathrm{Sym}^2}(a) &= \frac{1}{2} \left( 0 + 2 \right) = 1 \, ,
    \nonumber \\[2mm]
    \chi_{\mathrm{Alt}^2}(a) &= \frac{1}{2} \left( 0 - 2 \right) = -1 \, .
\end{align}
The only possible solution is that
\begin{equation}
    (\mathbf{1} \, \oplus \, \mathbf{1}')^{\otimes 2} = 
    \mathrm{Sym}^2(\mathbf{1} \, \oplus \, \mathbf{1}') \, \oplus \,
    \mathrm{Alt}^2(\mathbf{1} \, \oplus \, \mathbf{1}') =
    (\mathbf{1} \, \oplus \, \mathbf{1} \, \oplus \, \mathbf{1}') 
    \, \oplus \, (\mathbf{1}') \, .
\end{equation}
Consequently, the matrix $z_{ij,kl}$ decomposes as
\begin{align}\label{eq:2hdm_z2_z}
    z_{ij,kl} &\rightarrow \left[ (\mathbf{1} \, \oplus \, \mathbf{1} \, \oplus \, \mathbf{1}')
    \, \otimes \, (\mathbf{1} \, \oplus \, \mathbf{1} \, \oplus \, \mathbf{1}')  \right]
    \, \oplus \, 
    \left[ \mathbf{1}' \, \otimes \, \mathbf{1}' \right]
    \nonumber \\[2mm]
    &= 6(\mathbf{1}) \, \oplus \, 4(\mathbf{1}') \, .
\end{align}
Hence, we conclude that the 2HDM has $8$ parameters left after imposing
$\mathbb{Z}_2$, $2$ from $\mu_{ij}$ and $6$ from $z_{ij,kl}$. We note four important
facts. First, these $8$ parameters are not physical, but $7$ will be.
Indeed, $\lambda_5$ can be made real after rephasing.
Second, this result is completely basis-invariant. We might have chosen
an equivalent matrix other than $\mathrm{diag} (1,-1)$, but nevertheless, any equivalent
two dimensional
representation would still decompose as above.
Third, we have not
taken rephasings into account as this is a consequence of the action of a global $\mathrm{U}(1)$
for one field alone. This is only possible after $\mathbb{Z}_2$ is imposed.
Lastly, we may even do better than just count the number of parameters in $z_{ij,kl}$. 
By knowing how many singlets come from the $\mathrm{Alt}$ part, in this case just one,
we can use the decomposition in $\mathrm{SU}(2)$ to assign it to a particular group
of parameters. Similarly, we can do the same for the $\mathrm{Sym}$ part.
\end{remark}

\begin{remark}
We consider the 2HDM with a $\mathbb{Z}_3$ symmetry. We define $\omega = \exp(2i \pi /3)$
and choose $\mathbf{r} = \mathbf{1}' \, \oplus \, \mathbf{1}''$ corresponding
to the action of $\mathrm{diag}(\omega, \omega^2)$. The product rules are
given by
\begin{align}
    \mathbf{1}' \, \otimes \, \mathbf{1}' = \mathbf{1}'' \, , \quad
    \mathbf{1}'' \, \otimes \, \mathbf{1}'' = \mathbf{1}' \, , \quad
    \mathbf{1}' \, \otimes \, \mathbf{1}'' = \mathbf{1} \, .
\end{align}
Then,
\begin{equation}
    \mu_{ij} \rightarrow (\mathbf{1}' \, \oplus \, \mathbf{1}'')^{\otimes 2}
    = 2(\mathbf{1}) \, \oplus \, \mathbf{1}' \, \oplus \, \mathbf{1}'' \, .
\end{equation}
Choosing the element $g=a$ with $a^3 = e$, the identity element, we find
\begin{align}
    \chi_{\mathrm{Sym}^2}(a) &= \frac{1}{2} \left( 1 + (-1) \right) = 0 \, ,
    \nonumber \\[2mm]
    \chi_{\mathrm{Alt}^2}(a) &= \frac{1}{2} \left( 1 - (-1) \right) = 1 \, ,
\end{align}
where we used that $\chi(a) = \Tr \left[ \mathrm{diag}(\omega, \omega^2) \right] = -1
= \chi(a^2)$. Then, the only possibility is
\begin{equation}
    (\mathbf{1}' \, \oplus \, \mathbf{1}'')^{\otimes 2} = 
    \mathrm{Sym}^2(\mathbf{1}' \, \oplus \, \mathbf{1}'') \, \oplus \,
    \mathrm{Alt}^2(\mathbf{1}' \, \oplus \, \mathbf{1}'') =
    (\mathbf{1} \, \oplus \, \mathbf{1}' \, \oplus \, \mathbf{1}'') 
    \, \oplus \, (\mathbf{1}) \, .
\end{equation}
Therefore, we decompose $z_{ij,kl}$ as
\begin{align}
    z_{ij,kl} &\rightarrow \left[ (\mathbf{1} \, \oplus \, \mathbf{1}' \, \oplus \, \mathbf{1}'')
    \, \otimes \, (\mathbf{1} \, \oplus \, \mathbf{1}' \, \oplus \, \mathbf{1}'')  \right]
    \, \oplus \, 
    \left[ \mathbf{1} \, \otimes \, \mathbf{1} \right]
    \nonumber \\[2mm]
    &= 4(\mathbf{1}) \, \oplus \, 3(\mathbf{1}') \, \oplus \, 3(\mathbf{1}'') \, .
\end{align}
Hence, we conclude that the 2HDM with $\mathbb{Z}_3$ symmetry has $6$ parameters,
$2$ from $\mu_{ij}$ and $4$ from $z_{ij,kl}$. Here, the number of parameters
coincides with the number of physical parameters.
\end{remark}
\begin{remark}
We consider the 2HDM with a $U(1)$ symmetry. We use the transformation
$\mathrm{diag}(e^{i \xi}, e^{-i \xi})$, which corresponds to 
$\mathbf{r} = \bar{\mathbf{1}}' \, \oplus \, \mathbf{1}'$. The product rules are
given by
\begin{align}\label{eq:u1_prod_rules}
    \mathbf{1}' \, \otimes \, \mathbf{1}' = \mathbf{1}'' \, , \quad
    \bar{\mathbf{1}}' \, \otimes \, \bar{\mathbf{1}}' = \bar{\mathbf{1}}'' \, , \quad
    \bar{\mathbf{1}}' \, \otimes \, \mathbf{1}' = \mathbf{1} \, ,
\end{align}
where contrarily to the other examples, we have another representation
appearing in the product rules, a consequence
of $G$ infinite. With it, $\mu_{ij}$ decomposes as
\begin{equation}
    \mu_{ij} \rightarrow (\bar{\mathbf{1}}' \, \oplus \, \mathbf{1}')^{\otimes 2}
    = 2(\mathbf{1}) \, \oplus \, \mathbf{1}'' \, \oplus \, \bar{\mathbf{1}}'' \, .
\end{equation}
The characters, if we choose $a$ to be the element
with representation $\mathrm{diag}(e^{i \xi}, e^{-i \xi})$, are given by
\begin{align}
    \chi_{\mathrm{Sym}^2}(a) &= \frac{1}{2} 
    \left( 4 \cos ^2(\xi ) + 2 \cos (2 \xi ) \right) = 1 + 2 \cos(2 \xi) \, ,
    \nonumber \\[2mm]
    \chi_{\mathrm{Alt}^2}(a) &= \frac{1}{2} \left( 4 \cos ^2(\xi ) - 2 \cos (2 \xi ) \right) = 1 \, .
\end{align}
Thus, the only solution is given by
\begin{equation}\label{eq:u1_sym_alt}
    (\bar{\mathbf{1}}' \, \oplus \, \mathbf{1}')^{\otimes 2} = 
    \mathrm{Sym}^2(\bar{\mathbf{1}}' \, \oplus \, \mathbf{1}') \, \oplus \,
    \mathrm{Alt}^2(\bar{\mathbf{1}}' \, \oplus \, \mathbf{1}') =
    (\mathbf{1} \, \oplus \, \mathbf{1}'' \, \oplus \, \bar{\mathbf{1}}'') 
    \, \oplus \, (\mathbf{1}) \, .
\end{equation}
Therefore, we decompose $z_{ij,kl}$ as
\begin{align}
    z_{ij,kl} &\rightarrow \left[ (\mathbf{1} \, \oplus \, 
    \mathbf{1}'' \, \oplus \, \bar{\mathbf{1}}'')
    \, \otimes \, (\mathbf{1} \, \oplus \, 
    \mathbf{1}'' \, \oplus \, \bar{\mathbf{1}}'')  \right]
    \, \oplus \, 
    \left[ \mathbf{1} \, \otimes \, \mathbf{1} \right]
    \nonumber \\[2mm]
    &= 4(\mathbf{1}) \, \oplus \, 3(\mathbf{1}'''') \, \oplus \, 3(\bar{\mathbf{1}}'''') \, .
\end{align}
This result is remarkably similar to the case of $\mathbb{Z}_3$. This is not
coincidental as in fact they lead to the same symmetry constraint
in the 2HDM \cite{Ferreira:2008zy}. This can be
seen from the fact that it only differs in the representations that are primed,
those that we will not keep.
\end{remark}
\begin{remark}
We consider the 2HDM with a $S_3$ symmetry. We choose $\mathbf{r} = \mathbf{2}$ corresponding
to the action of doublet representation in the fields. The product rules are
given by
\begin{align}\label{eq:s3_prod_rules}
    \mathbf{1}' \, \otimes \, \mathbf{1}' = \mathbf{1} \, , \quad
    \mathbf{1}' \, \otimes \, \mathbf{2} = \mathbf{2} \, , \quad
    \mathbf{2} \, \otimes \, \mathbf{2} = \mathbf{1} \, \oplus \, \mathbf{1}' \, \oplus \,
    \mathbf{2} \, .
\end{align}
Then $\mu_{ij}$ decomposes as
\begin{equation}
    \mu_{ij} \rightarrow \mathbf{2}^{\otimes 2}
    = \mathbf{1} \, \oplus \, \mathbf{1}' \, \oplus \, \mathbf{2} \, .
\end{equation}
Choosing $g=(1,2)$ and consulting the character table we get
\begin{align}
    \chi_{\mathrm{Sym}^2}(g) &= \frac{1}{2} 
    \left( 0 + 2 \right) = 1 \, ,
    \nonumber \\[2mm]
    \chi_{\mathrm{Alt}^2}(g) &= \frac{1}{2} \left( 0 - 2 \right) = -1 \, .
\end{align}
Then, the only possibility is
\begin{equation}\label{eq:s3_sym_alt}
    \mathbf{2}^{\otimes 2} = 
    \mathrm{Sym}^2(\mathbf{2}) \, \oplus \,
    \mathrm{Alt}^2(\mathbf{2}) =
    (\mathbf{1} \, \oplus \, \mathbf{2}) 
    \, \oplus \, (\mathbf{1}') \, .
\end{equation}
Thus, $z_{ij,kl}$ decomposes as
\begin{align}
    z_{ij,kl} &\rightarrow \left[ (\mathbf{1} \, \oplus \, 
    \mathbf{2})
    \, \otimes \, (\mathbf{1} \, \oplus \, 
    \mathbf{2})  \right]
    \, \oplus \, 
    \left[ \mathbf{1}' \, \otimes \, \mathbf{1}' \right]
    \nonumber \\[2mm]
    &= 3(\mathbf{1}) \, \oplus \, \mathbf{1}' \, \oplus \, 3(\mathbf{2}) \, .
\end{align}
Consequently, this model has $4$ parameters, $1$ from $\mu_{ij}$ and $3$ from $z_{ij,kl}$.
This can be checked against \cite{Cogollo:2016dsd}.
\end{remark}
\begin{remark}
We consider the 3HDM with a $A_4$ symmetry. We choose $\mathbf{r} = \mathbf{3}$ corresponding
to the action of doublet representation in the fields. The product rules are
given by
\begin{align}\label{eq:a4_prod_rules}
    \mathbf{1}' \, \otimes \, \mathbf{1}' = \mathbf{1}'' \, , \quad
    \mathbf{1}'' \, \otimes \, \mathbf{1}'' = \mathbf{1}' \, , \quad
    \mathbf{1}' \, \otimes \, \mathbf{1}'' = \mathbf{1} \, , \quad
    \mathbf{3} \, \otimes \, \mathbf{3} = \mathbf{1} \, \oplus \, \mathbf{1}' \, \oplus \,
    \mathbf{1}'' \, \oplus \, 2(\mathbf{3}) \, .
\end{align}
Then $\mu_{ij}$ decomposes as
\begin{equation}
    \mu_{ij} \rightarrow \mathbf{3}^{\otimes 2}
    = \mathbf{1} \, \oplus \, \mathbf{1}' \, \oplus \,
    \mathbf{1}'' \, \oplus \, 2(\mathbf{3}) \, .
\end{equation}
Choosing $g=(1,2)(3,4)$ and consulting the character table we get
\begin{align}
    \chi_{\mathrm{Sym}^2}(g) &= \frac{1}{2} 
    \left( (-1)^2 + 3 \right) = 2 \, ,
    \nonumber \\[2mm]
    \chi_{\mathrm{Alt}^2}(g) &= \frac{1}{2} \left( (-1)^2 - 3  \right) = -1 \, .
\end{align}
where we used $g^2 = e$, the identity element. The only possibility is
\begin{equation}\label{eq:a4_sym_alt}
    \mathbf{3}^{\otimes 2} = 
    \mathrm{Sym}^2(\mathbf{3}) \, \oplus \,
    \mathrm{Alt}^2(\mathbf{3}) =
    (\mathbf{1} \, \oplus \, \mathbf{1}' \, \oplus \, \mathbf{1}'' \, \oplus \, \mathbf{3}) 
    \, \oplus \, (\mathbf{3}) \, .
\end{equation}
Thus, $z_{ij,kl}$ decomposes as
\begin{align}
    z_{ij,kl} &\rightarrow \left[ (\mathbf{1} \, \oplus \, 
    \mathbf{1}' \, \oplus \, \mathbf{1}'' \, \oplus \, \mathbf{3}) \,
    \otimes \, (\mathbf{1} \, \oplus \, \mathbf{1}' \, \oplus \, 
    \mathbf{1}'' \, \oplus \, \mathbf{3}) \right]
    \, \oplus \, 
    \left[ \mathbf{3} \, \otimes \, \mathbf{3} \right]
    \nonumber \\[2mm]
    &= 5(\mathbf{1}) \, \oplus \, 5(\mathbf{1}')
    \, \oplus \, 5(\mathbf{1}'') \, \oplus \, 10(\mathbf{3}) \, .
\end{align}
With this we conclude that this model has $6$ parameters, $1$ from $\mu_{ij}$ 
and $5$ from $z_{ij,kl}$.
This can be checked against \cite{Ivanov:2014doa}.
\end{remark}
\begin{remark}
Finally we consider the 3HDM with a $S_4$ symmetry.
We will not
write every tensor product rule in this case they are many.
Choosing $\mathbf{r}=\mathbf{3}$ we get
\begin{equation}
    \mu_{ij} \rightarrow \mathbf{3}^{\otimes 2}
    = \mathbf{1} \, \oplus \, \mathbf{2} \, \oplus \,
    \mathbf{3} \, \oplus \, \mathbf{3}' \, .
\end{equation}
Choosing $g=(1,2)$ we get the characters
\begin{align}
    \chi_{\mathrm{Sym}^2}(g) &= \frac{1}{2} 
    \left( 1^2 + 3 \right) = 2 \, ,
    \nonumber \\[2mm]
    \chi_{\mathrm{Alt}^2}(g) &= \frac{1}{2} \left( 1^2 - 3  \right) = -1 \, ,
\end{align}
and then
\begin{align}
    z_{ij,kl} &\rightarrow \left[ (\mathbf{1} \, \oplus \, 
    \mathbf{2} \, \oplus \, \mathbf{3}) \,
    \otimes \, (\mathbf{1} \, \oplus \, 
    \mathbf{2} \, \oplus \, \mathbf{3}) \right]
    \, \oplus \, 
    \left[ \mathbf{3}' \, \otimes \, \mathbf{3}' \right]
    \nonumber \\[2mm]
    &= 4(\mathbf{1}) \, \oplus \, \mathbf{1}'
    \, \oplus \, 5(\mathbf{2}) \, \oplus \, 6(\mathbf{3}) \, \oplus \, 4(\mathbf{3}') \, .
\end{align}
Thus, the 3HDM with $S_4$ symmetry has $5$ parameters,
$1$ from $\mu_{ij}$ and $4$ from $z_{ij,kl}$. This can be checked against
\cite{Ivanov:2014doa}.
\end{remark}
There are many interesting analysis that one can make from these examples.
It is easy to check that choosing the 2HDM with a symmetry $\mathrm{diag}(i, -i)$
yields a similar result to the one we obtained with $\mathbb{Z}_2$. However,
choosing $\mathrm{diag}(1, i)$ yields a similar result to the one of $\mathbb{Z}_3$
and $\mathrm{U}(1)$. Both are cases in which a group effectively acts as another.

One result we can infer from this and previous sections is that due to the decomposition
of any NHDM into representations of $\mathrm{SU}(N)$, we are always guaranteed to have
three singlets. This result follows from the fact that any $G$ that we choose will
be a subgroup of $\mathrm{PSU}(N)$. If $\mathrm{SU}(N)$ guarantees three singlets,
so will any symmetry groups. Furthermore, these three will be physical parameters.

\section{Conclusions}

We studied in detail the group structure of the matrices in the scalar potential
of multi-Higgs doublet models. We show its decomposition under
irreducible representations of $\mathrm{SU}(N)$ with a simple formula
using the symmetric and antisymmetric part of the tensor product. With
this decomposition, the study of the physical parameters of the theory
becomes attainable.

We have used a tool from partition theory, Omega calculus, to
compute complicated Hilbert series without using the residue theorem. 
Its use in high-energy physics is a first, as most computations depend on the
residue theorem.
In particular,
we compute for the first time the closed form of the Hilbert series of the 3HDM,
a result previously very difficult to obtain by standard methods. From this
function, we will be able to completely characterize the physical parameters
of the 3HDM.

Using a number of formal results in invariant theory we proved that the most
general NHDM has $(N^4 + N^2 +2)/2$ physical parameters. We also showed
a theorem on the decomposition of NHDM into irreducible representations
of $\mathrm{SU}(N)$.
We presented a formula to decompose the matrices of the Lagrangian for all
$N>3$.

For the first time we derived a basis-invariant method for
counting parameters in a Lagrangian with both
basis-invariant redundancies and global symmetries. We show that
the knowledge of tensor product decomposition and character theory
is enough for attaining this purpose. Furthermore, this technique does
not require analysis of the Lagrangian itself.

With invariant theory, we hope that a clear path to a full basis-invariant
overview to the physical parameters of NHDM will soon be possible.
There are still many unanswered questions on the CP properties and the
physical parameters in theories with symmetries,
both of which we have not addressed in this paper.

\vspace{2ex}

\acknowledgments
\noindent
M. P. B. is very grateful to J. P. Silva for all the useful discussions on
scalar models
and endless advice. M. P. B. is also grateful to A. Trautner for
the time spent explaining
invariant theory in scalar models.
This work is supported in part by the Portuguese
Funda\c{c}\~{a}o para a Ci\^{e}ncia
e Tecnologia (FCT) under contract SFRH/BD/146718/2019.
This work is also supported in part by
FCT under contracts
CERN/FIS-PAR/0008/2019,
PTDC/FIS-PAR/29436, UIDB/00777/2020,
and UIDP/00777/2020.




\appendix

\section{Proof of theorem~\ref{th:rf}}\label{sec:appA}
 
In section~\ref{sec:group_s} we decomposed
the matrices $\mu$ and $z$ of the scalar potential
\begin{equation}
    V_{H} = \mu_{ij} (\Phi^\dagger_i \Phi_j) + z_{ij,kl} (\Phi^\dagger_i \Phi_j) (\Phi^\dagger_k \Phi_l) \, ,
\end{equation}
such that its bare elements transform under a direct sum of irreducible representations of $\mathrm{SU}(N)$.
In order to do that we have followed Trautner \cite{Trautner:2018ipq} in the use of projection
operators
to enforce hermiticity and symmetrization in indices. 
With it we concluded that the matrices $\mu$ and $z$ decompose as
\begin{align}\label{eq:decommuz}
    &\mu_{i j} \rightarrow \bar{\mathbf{r}}  \, \otimes \, \mathbf{r} \, , \nonumber \\[2mm]
    &z_{ij,kl} \rightarrow \left[ \mathrm{Sym}^2 (\mathbf{r}) \,
    \otimes \, \mathrm{Sym}^2 (\bar{\mathbf{r}} ) \right] \, \oplus \, 
    \left[ \mathrm{Alt}^2 (\mathbf{r}) \, \otimes \, \mathrm{Alt}^2 (\bar{\mathbf{r}} )  \right] \, .
\end{align}
With eq.~\eqref{eq:decommuz} we provide a proof of theorem~\ref{th:rf}.

Knowing that the decomposition of $\mu$ is trivial, we will focus our attention
to the decomposition of $z_{ij,kl}$. To that end, we make use of Young tableaux.
Let $G=\mathrm{SU}(N)$ with $N>3$. Then the
square of the fundamental and anti-fundamental representations
is given by
\ytableausetup{centertableaux}
\begin{equation}
\mathbf{r}_f \, \otimes \, \mathbf{r}_f =
\begin{ytableau}
    ~
\end{ytableau}
\, \otimes \,
\begin{ytableau}
    ~
\end{ytableau}
=
\begin{ytableau}
    ~ \\
    ~
\end{ytableau}
\, \oplus \,
\begin{ytableau}
    ~ & ~
\end{ytableau}
= \mathrm{Alt}^2(\mathbf{r}_f) \, \oplus \, \mathrm{Sym}^2(\mathbf{r}_f) \, ,
\end{equation}
and
\begin{equation}
\bar{\mathbf{r}}_f \, \otimes \, \bar{\mathbf{r}}_f =
{}_{N-1} \left\{
\begin{ytableau}
    ~ \\
    \vdots \\
    ~ \\
\end{ytableau}
\right.
\, \otimes \,
\begin{ytableau}
    ~ \\
    \vdots \\
    ~ \\
\end{ytableau}
=
{}_{N-2} \left\{
\begin{ytableau}
    ~ \\
    \vdots \\
    ~
\end{ytableau}
\right.
\, \oplus \,
{}_{N-1} \left\{
\begin{ytableau}
    ~ & ~ \\
    \vdots & \vdots \\
    ~ & ~ 
\end{ytableau}
\right.
= \mathrm{Alt}^2(\bar{\mathbf{r}}_f) \, \oplus \, \mathrm{Sym}^2(\bar{\mathbf{r}}_f) \, .
\end{equation}
Thus, we compute the terms in eq.~\eqref{eq:decommuz} as
\begin{align}\label{eq:young_symsym}
\mathrm{Sym}^2(\bar{\mathbf{r}}_f) \, \otimes \, \mathrm{Sym}^2(\mathbf{r}_f) &=
{}_{N-1} 
\left\{
\begin{ytableau}
    ~ & ~ \\
    \vdots & \vdots \\
    ~ & ~ 
\end{ytableau}
\right.
\, \otimes \,
\begin{ytableau}
    ~ & ~
\end{ytableau}
\nonumber \\[5mm]
&=
\mathbf{1} \, \oplus \,
{}_{N - 1}  \left\{
\begin{ytableau}
    ~ & ~ \\
    \vdots & \none \\
    ~
\end{ytableau}
\right.
\, \oplus \,
{}_{N-1} 
\left\{
\begin{ytableau}
    ~ & ~ & ~ & ~ \\
    \vdots & \vdots & \none \\
    ~ & ~ & \none
\end{ytableau}
\right. \, ,
\end{align}
and also
\begin{align}\label{eq:young_altalt}
\mathrm{Alt}^2(\bar{\mathbf{r}}_f) \, \otimes \, \mathrm{Alt}^2(\mathbf{r}_f) &=
{}_{N-2} 
\left\{
\begin{ytableau}
    ~ \\
    \vdots \\
    ~
\end{ytableau}
\right.
\, \otimes \,
\begin{ytableau}
    ~ \\
    ~
\end{ytableau}
\nonumber \\[5mm]
&=
\mathbf{1} \, \oplus \,
{}_{N - 1}  \left\{
\begin{ytableau}
    ~ & ~ \\
    \vdots & \none \\
    ~
\end{ytableau}
\right.
\, \oplus \,
{}_{N-2} 
\left\{
\begin{ytableau}
    ~ & ~ \\
    \vdots & ~ \\
    ~ & \none
\end{ytableau}
\right. \, .
\end{align}
It is clear that this procedure is valid for $N > 3$. For $N=3$ the last term
in eq.~\eqref{eq:young_altalt} does not exist, but the remaining terms do. 
Next,
we compute the dimensions of each individual term in
eqs.~\eqref{eq:young_symsym}--\eqref{eq:young_altalt} by using the known
formula for the dimension of Young tableaux in $\mathrm{SU}(N)$. Then, we have
\begin{equation}
\dim \left( 
{}_{N - 1}  \left\{
\begin{ytableau}
    ~ & ~ \\
    \vdots & \none \\
    ~
\end{ytableau}
\right.
\right) = \frac{(N + 1) N!}{N (N-2)!} = N^2 - 1 \, ,
\end{equation}
the adjoint representation. The dimension of the last term of eq.~\eqref{eq:young_symsym}
is given by 
\begin{equation}
\dim \left( 
{}_{N-1} 
\left\{
\begin{ytableau}
    ~ & ~ & ~ & ~ \\
    \vdots & \vdots & \none \\
    ~ & ~ & \none
\end{ytableau}
\right.
\right) = \frac{(N + 2) (N + 3) N! (N + 1)!}{4 (N+1)(N+2)(N-1)!(N-2)!} =
\frac{N^2(N-1)(N+3)}{4} \, ,
\end{equation}
and the dimension of the last term of eq.~\eqref{eq:young_altalt}
is given by 
\begin{equation}
\dim \left( 
{}_{N-2} 
\left\{
\begin{ytableau}
    ~ & ~ \\
    \vdots & ~ \\
    ~ & \none
\end{ytableau}
\right.
\right) = \frac{N (N + 1)!}{4 (N-2)(N-1)(N-4)!} =
\frac{N^2(N+1)(N-3)}{4} \, ,
\end{equation}
where we already recognize the sequences that we presented in table~\ref{table:nhdm_reps}
as
\begin{equation}
    a_N = \frac{N^2(N+1)(N-3)}{4} \quad \text{and} \quad b_N = \frac{N^2(N-1)(N+3)}{4} \, .
\end{equation}
Then, the full decomposition is given by
\begin{align}
    V = 3(\mathbf{1}) \, \oplus \, 3\left( \mathbf{N^2-1} \right)  \, 
    \oplus \, \left( \mathbf{\frac{N^2(N+1)(N-3)}{4}} \right) 
    \, \oplus \, \left( \mathbf{\frac{N^2(N-1)(N+3)}{4}} \right) \, ,
\end{align}
thus completing the proof.

\section{Hilbert series in $\mathrm{SU}(3)$}\label{sec:appB}

For completeness and such that we may provide results for the reader in case
there is need for Hilbert series in $\mathrm{SU}(3)$ we list the ungraded Hilbert series
that we computed before going to the case of the 3HDM.

\subsection{One $\mathbf{8}$}

For the case of one $\mathbf{8}$, $\dim V = 8$ and the Hilbert series
is given by
\begin{equation}
    H(K[V]^G, t) = \frac{1}{(1 - t^2) (1 - t^3)} \, ,
\end{equation}
where the Krull dimension is given by $\dim K[V]^G = 2$.

\subsection{Two $\mathbf{8}$'s}

For the case of two $\mathbf{8}$'s, $\dim V = 16$ and the Hilbert series
is given by
\begin{equation}
    H(K[V]^G, t) = \frac{1+t^6}{(1 - t^2)^3 (1 - t^3)^4 (1 - t^4)} \, ,
\end{equation}
where the Krull dimension is given by $\dim K[V]^G = 8$.

\subsection{Three $\mathbf{8}$'s}

For the case of three $\mathbf{8}$'s, $\dim V = 24$ and the Hilbert series
is given by
\begin{equation}
    H(K[V]^G, t) = \frac{1 + 3 t^3 + 7 t^4 + 9 t^5 + 16 t^6 + 18 t^7 + 25 t^8 + 30 t^9 + 
 34 t^{10} + \dots + t^{20}}
    {(1 - t^2)^6 (1 - t^3)^8 (1 - t^4)^2} \, ,
\end{equation}
where we omitted terms in the numerator but since it is palindromic, they are easy
to compute.
The Krull dimension is given by $\dim K[V]^G = 16$.

\subsection{One $\mathbf{27}$}

For the case of one $\mathbf{27}$, $\dim V = 27$ and the Hilbert series
is given by
\begin{equation}
    H(K[V]^G, t) = \frac{P(t)}
    {(1 - t^2) (1 - t^3)^2 (1 - t^4)^3 (1 - t^5)^4 (1 - t^6)^5 (1 - 
   t^7)^2 (1 - t^8) (1 - t^9)} \, ,
\end{equation}
where the numerator is too large to show in this form.
It is given by a palindromic polynomial of degree $74$ for which we
list the first $37$ coefficients
\begin{align}
    \mathrm{Coefficients} = \{ & 1, 0, 0, 0, 0, 0, 6, 15, 34, 73, 139, 258, 482, 851, 1486, 2531,
4148, 6603, \nonumber \\
& 10222, 15334, 22377, 31836, 44133, 59736, 79024, 102166, 129198, \nonumber \\
& 159916, 193698, 229724, 266860, 303653, 338555, 369956, 396288, \nonumber \\
&  416179, 428567, 432774, 428567, \dots \} \, .
\end{align}

The Krull dimension is given by $\dim K[V]^G = 19$.

\subsection{One $\mathbf{27}$ and one $\mathbf{8}$}

For the case of one $\mathbf{27}$ and one $\mathbf{8}$,
$\dim V = 35$ and the Hilbert series
is given by
\begin{equation}
    H(K[V]^G, t) = \frac{P(t)}
    {(1 - t^2)^2 (1 - t^3)^5 (1 - t^4)^6 (1 - t^5)^5 (1 - t^6)^4 (1 - 
   t^7)^3 (1 - t^8) (1 - t^9)} \, ,
\end{equation}
where the numerator is again too large to show in this form.
It is given by a palindromic polynomial of degree $95$ for which we
list the first $48$ coefficients
\begin{align}
    \mathrm{Coefficients} = \{ & 1, 0, 0, 0, 3, 18, 67, 177, 486, 1257, 3124, 7514, 17381, 38427,
81953, 168322, \nonumber \\
& 333782, 640599, 1191529, 2150336, 3771546, 6434476, 10689459,  \nonumber \\
& 17309116, 27342618, 42168281, 63541423, 93612205, 134923454, \nonumber \\
& 190359277, 263038374, 356147162, 472718380, 615332656, \nonumber \\
& 785802717, 984828515, 1211667639, 1463867914, 1737104026, \nonumber \\
& 2025120038, 2319866047, 2611789299, 2890287579, 3144319605, \nonumber \\
& 3363113113, 3536884827, 3657565340, 3719405182, \dots \} \, .
\end{align}

The Krull dimension is given by $\dim K[V]^G = 27$.

\subsection{One $\mathbf{27}$ and two $\mathbf{8}$'s}

For the case of one $\mathbf{27}$ and two $\mathbf{8}'s$,
$\dim V = 43$ and the Hilbert series
is given by
\begin{equation}
    H(K[V]^G, t) = \frac{P(t)}
    {\left(1-t^2\right)^4 \left(1-t^3\right)^7 \left(1-t^4\right)^7 \left(1-t^5\right)^7 
    \left(1-t^6\right) \left(1-t^7\right)^4 \left(1-t^9\right) \left(1-t^{12}\right)^4} \, ,
\end{equation}
where the numerator is again too large to show in this form.
It is given by a palindromic polynomial of degree $140$ for which we
list the first $71$ coefficients
\begin{align}
    \mathrm{Coefficients} = \{ & 1, 0, 0, 4, 19, 86, 345, 1146, 3827, 12155, 36644, 105650, 291364, \nonumber \\[2mm]
& 768516, 1948541, 4755476, 11193876, 25464839, 56078682, 119723055, \nonumber \\[2mm]
& 248171206, 500124587, 981064230, 1875483402, 3497777483, 6370373651, \nonumber \\[2mm]
& 11340751085, 19751538630, 33682079132, 56282123455, 92221097955, \nonumber \\[2mm]
& 148276916468, 234088220042, 363086857757, 553623088130, 830279318916, \nonumber \\[2mm]
& 1225352927988, 1780463739188, 2548218933862, 3593828962209, \nonumber \\[2mm]
& 4996550018628, 6850786076060, 9266669669316, 12369917820989, \nonumber \\[2mm]
& 16300770543819, 21211810440177, 27264530683140, 34624537424009, \nonumber \\[2mm]
& 43455377357385, 53911052118432, 66127418414828, 80212736271900, \nonumber \\[2mm]
& 96237813744402, 114226219326836, 134145178588251, 155897769157145, \nonumber \\[2mm]
& 179317092968054, 204162976726575, 230121790218424, 256809697567690, \nonumber \\[2mm]
& 283779596203437, 310531697108254, 336527558907554, 361207024717583, \nonumber \\[2mm]
& 384007472002695, 404384397694558, 421832396805165, 435905386086712, \nonumber \\[2mm]
& 446235036730978, 452546292870718, \dots \} \, .
\end{align}

The Krull dimension is given by $\dim K[V]^G = 35$.

\subsection{The 3HDM --- one $\mathbf{27}$ and three $\mathbf{8}$'s}\label{app:B_3hdm}

For the case of one $\mathbf{27}$ and three $\mathbf{8}'s$,
$\dim V = 51$ and the Hilbert series
is given by
\begin{equation}
    H(K[V]^G, t) = \frac{P(t)}
    {\left(1-t^2\right)^7 \left(1-t^3\right)^8 \left(1-t^4\right)^6 \left(1-t^5\right)^9 
    \left(1-t^6\right)^3 \left(1-t^7\right)^5 \left(1-t^9\right) \left(1-t^{12}\right)^4} \, ,
\end{equation}
where the numerator is again too large to show in this form.
It is given by a palindromic polynomial of degree $166$ for which we
list the first $84$ coefficients
\begin{align}
    \mathrm{Coefficients} = \{ & 1, 0, 0, 14, 60, 275, 1274, 5155, 20161, 75095, 264240, 885516, \nonumber \\[2mm]
& 2834022, 8671076, 25445735, 71779143, 195007048, 511236133, \nonumber \\[2mm]
& 1295636355, 3179224601, 7564677192, 17477956901, 39262235498, \nonumber \\[2mm]
& 85853683736, 182945453610, 380283856670, 771851230696, 1531042093644, \nonumber \\[2mm]
& 2970515462995, 5641679579607, 10496317684056, 19143367022397, \nonumber \\[2mm]
& 34248056625345, 60139123080143, 103713241781054, 175754293393734, \nonumber \\[2mm]
& 292818290150518, 479871564512319, 773908102619016, 1228807989461265, \nonumber \\[2mm]
& 1921728177626469, 2961332582100456, 4498163108992467, \nonumber \\[2mm]
& 6737417725756358, 9954327329157130, 14512200512939650, \nonumber \\[2mm]
& 20883028531401563, 29670307379183676, 41633441343039786, \nonumber \\[2mm]
& 57712749259880728, 79053716812829777, 107028745453572029, \nonumber \\[2mm]
& 143254274314089003, 189600830977724482, 248193344335470998, \nonumber \\[2mm]
& 321398971394200935, 411799796686223825, 522148088845635419, \nonumber \\[2mm]
& 655302376862169534, 814143442108891869, 1001470404161839701, \nonumber \\[2mm]
& 1219878370055240399, 1471620565854742943, 1758459380263239671, \nonumber \\[2mm]
& 2081512232982500064, 2441099505001154860, 2836602817597282373, \nonumber \\[2mm]
& 3266342593242487553, 3727483983236093502, 4215979803568794437, \nonumber \\[2mm]
& 4726558065572040847, 5252760002583348836, 5787032245088490827, \nonumber \\[2mm]
& 6320874073411653649, 6845037648170077207, 7349775938749897776, \nonumber \\[2mm]
& 7825129987928475739, 8261244356119359550, 8648697317481959007, \nonumber \\[2mm]
& 8978830801547884154, 9244064360275753923, 9438177639914705696, \nonumber \\[2mm]
& 9556547018948486707, 9596324149303158505, \dots \} \, .
\end{align}

The Krull dimension is given by $\dim K[V]^G = 43$.

\bibliographystyle{JHEP}
\bibliography{bibliography}

\end{document}